\documentclass[runningheads]{llncs}
\usepackage{amsmath}

\usepackage[T1]{fontenc}
\usepackage{graphicx}
\usepackage{times} 
\usepackage{helvet} 
\usepackage{courier} 
\usepackage[hyphens]{url} 
\usepackage[square,sort,comma,numbers]{natbib} 

\usepackage{enumitem}

\usepackage[utf8]{inputenc} 
\usepackage[T1]{fontenc}    
\usepackage{url}            
\usepackage{booktabs}       
\usepackage{amsfonts}       
\usepackage{nicefrac}       
\usepackage{microtype}      
\usepackage{xcolor}    

\usepackage{multirow,array}

\graphicspath{{images/}}
\usepackage{amssymb}
\usepackage{csquotes}
\usepackage{color}
\usepackage{thm-restate}
\usepackage{subcaption}
\usepackage{tikz}
\usepackage{pgfplots}

\pgfplotsset{compat=1.18}

\begin{document}
\newif\ifcomments   

\commentstrue

\newcommand{\yotam}[1]{{\ifcomments \color{red}{Yotam: #1} \fi }}
\newcommand{\ronen}[1]{{\ifcomments \color{blue}{Ronen: #1} \fi }}

\newcommand{\moshe}[1]{{\ifcomments \color{blue}{Moshe: #1} \fi }}

\newcommand{\agents}{\mathcal{N}} 
\newcommand{\ninput}{\mathcal{I}} 
\newcommand{\uu}{\bigcup_{U\in \mathbf{U^t}}U}
\newcommand{\upu}{\biguplus_{U\in \mathbf{U^t}}U}

\newcommand{\argmin}{\operatornamewithlimits{arg\ min}}

\newcommand{\nitest}[1]{\mathsf{NT}\left(#1\right)}

\newcommand{\tlabel}{t} 

\title{Prediction-sharing During Training and Inference}
%
%
\author{Yotam Gafni\inst{1}
\and
Ronen Gradwohl\inst{2}
\and
Moshe Tennenholtz\inst{3}
}
\authorrunning{Gafni, Gradwohl, and Tennenholtz}
%
\institute{Weizmann Institute 
\email{yotam.gafni@gmail.com}\and
Ariel University
\email{roneng@ariel.ac.il}\\
 \and
Technion - Israel Institute of Technology
\email{moshet@ie.technion.ac.il}}
\maketitle              

\begin{abstract}

Two firms are engaged in a competitive prediction task. Each firm has two sources of data---labeled historical data and unlabeled inference-time data---and uses the former to derive a prediction model, and the latter to make predictions on new instances. We study data-sharing contracts between the firms. The novelty of our study is to introduce and highlight the differences between contracts that share prediction models only, contracts to share inference-time predictions only, and contracts to share both.

Our analysis proceeds on three levels. First, we develop a general Bayesian framework that facilitates our study.
Second, we narrow our focus to two natural settings within this framework: (i) a  setting in which the accuracy of each firm's prediction model is common knowledge, but the correlation between the respective models is unknown; and (ii) a setting in which two hypotheses exist regarding the optimal predictor, and one of the firms has a structural advantage in deducing it. 

Within these two settings we study optimal contract choice. More specifically, we find the individually rational and Pareto-optimal contracts for some notable cases, and describe specific settings where each of the different sharing contracts emerge as optimal. Finally, in the third level of our analysis we demonstrate the applicability of our concepts in a synthetic simulation using real loan data. 

\end{abstract}

\section{Introduction}
Machine learning (ML) is becoming a highly distributed endeavor. Data is spread among different firms, each of whom may have their own ML capabilities and economic utilities. In many cases, one firm's data and prediction capabilities are complemented by those available to a competing firm, and each firm would benefit from access to the other's  predictions. 
 For example, two investment banks that attempt to predict loan defaults could each improve their respective predictions by accessing the other's predictions. Indeed, this is in the spirit of one of the most fundamental ideas in ML---aggregating weak learners into strong ones \cite{FreundS97}. However, the distributed nature of firms' capabilities introduces a major obstacle:  Why, and under what conditions, would firms willingly share their predictions with competitors? And what would equilibrium behavior look like, given such sharing? 

Our main innovation in this paper is the observation that this obstacle actually consists of two separate questions, corresponding respectively to the training and inference phases in ML. First,
why would firms share the labels they have (about individuals in their data logs) in the training phase? And second,  why would firms  share their predictions (about new, incoming instances) in the inference phase? As we show in this paper, this distinction has real bite. 

In order to tackle the question of training/inference-stage prediction-sharing, we proceed on three levels. First, we develop a general Bayesian model that captures the two kinds of sharing. The Bayesian model specifies the informational environment, while a utility model specifies the economic implications. In the Bayesian model, each firm obtains a {\em training signal} that represents the {\em prediction model} (a.k.a.\ classifier) learned by that firm via its labeled historical data. The firm also obtains an {\em inference-time signal} that represents the classifier's prediction on unlabeled inference-time data. In the utility model we associate a real number with each outcome quadrant: True-positive, true-negative, false-positive, and false-negative predictions. We moreover assume that if both firms arrive at the same outcome, then the associated utility is split between them.  

In the second level of our analysis, we apply our model to a game-theoretic study of two natural settings. 
In the first setting, the accuracy of each firm's prediction model is common knowledge, but the correlation between the respective models is unknown. As for utilities, each firm has a safe prediction that yields utility zero (whether right or wrong), and a risky prediction. For example, a firm predicting a customer's trustworthiness in order to decide whether or not to issue a loan.  If a loan is provided, the firm's utility depends on the accuracy of the trustworthiness prediction, and whether or not the customer has other offers. If no loan is provided, the firm's utility is fixed at $0$.
In the second setting we study, there are two hypotheses regarding the optimal predictor, and one of the firms has a structural advantage in deriving it. Furthermore, firms' utilities are symmetric across prediction types (unlike the first setting), and depend only on the predictions' correctness---e.g., a firm recommending a movie to a viewer, where the utility depends on whether or not it accurately predicts the viewer's tastes.

Finally, in the third level of our analysis, we demonstrate the applicability of our ideas in a synthetic simulation using real loan data. This is intended to provide an accessible, practical recasting of our abstract model's results. In broad terms, if we take a single firm's perspective, the \textit{no-sharing} contract allows it to build a classifier based on its own historical data. Then, based on its assessment (prior) of the competitor, it decides whether or not to act in accordance with the classifier's prediction (signal). An example of choosing to ignore the classifier's signal would be if the firm knows that its competitor can perfectly predict whether a loan would be repaid. Then, all the benefit of issuing a good loan is split (e.g., by the random decision of the consumer as to which of the offered loans to accept). However, since the firm knows that its own classifier is imperfect, it knows it will also end up issuing some bad loans. If the cost of bad loans outweighs the benefit of splitting the profit from good loans, the firm would decide to ignore its classifier and not issue any loans. Expanding on this example, the \textit{train-sharing} contract can allow the firm to make a more refined decision: Based on seeing how the other firm predicts on the historical data, it can assess whether or not to follow its own classifier. The \textit{full-sharing} contract allows even more intricate decision rules: They can depend both on what the firm learns about the competitor's predictions on historical data, and also on the competitor's prediction on each specific consumer. Lastly, the \textit{infer-sharing} contract does not see the competitor's predictions on historical data, so it must maintain its assessment/prior over the other firm's classifier, but it can use the competitor's prediction on the real-time consumer to decide whether to follow its own classifier's prediction. In our practical implementation of Section~\ref{sec:simulation} we examine the performance of the optimal decision rules under different contracts, and show that each of no-sharing, full-sharing, and train-sharing is uniquely optimal for some set of parameters. 


%


The emphasis of our work in game theoretic terms is to require that a contract is both \textit{individually rational} and \textit{Pareto-optimal} (IRPO). This follows the assumption that the natural state of affairs is that no contract is signed (no-sharing). Thus, for the firms to agree for any kind of prediction-sharing, it must be that for each of them, the expected utility under the prediction-sharing contract is at least as good as under no-sharing. We refer to this property as the contract being \textit{individually rational}. Moreover, the contract must be \textit{Pareto-optimal} w.r.t. the four possible contracts. E.g., if the utilities under full-sharing dominate these under train-sharing, even if train-sharing is by itself individually rational, it would make sense that the firms choose to sign the Pareto-optimal contract rather than a Pareto dominated one. As we see, there are different settings so that each of the contract types may become uniquely IRPO. 

Lastly, we note that in order for the firms to share their predictions, they need a way to match records. Facing this issue is common in the industry and there are companies that specialize in this task.\footnote{E.g., in advertising, identifying the same user on different devices is called cross-device targeting, and ``attribution providers'' companies such as AppsFlyer and Singular enable this.}
This type of prediction-sharing is valuable, even if done for identifiers both firms hold, as different firms may be exposed to different properties of the same identifier. As an example, think of firms that know different social and financial features associated with the same social security number. In this case, there is a difference between sharing each firm’s binary prediction regarding the user, or the entire data it holds for that identifier. Importantly, our model assumes that firms share  their training and inference-time \textit{signals}, and not their entire data. In practice, in the training stage the signals come in the form of  true labels in the historical data, and in the inference stage in the form of the classifier's predictions. 
The fact that this still proves to be useful is by itself interesting, as it suggests a path to data sharing that protects both the firm's intellectual property (in terms of both data and models used in training), and possibly the users' privacy.

\subsection{Our Contribution}
In Section~\ref{sec:model}, we provide the first model to reason about contracts that may involve sharing prediction both in the training and inference stage. In Sections~\ref{sec:corrModel},\ref{sec:two_hyp} we focus on two natural sub-models of the general model we present: 
\begin{enumerate}
 
 \item \textit{A Correlation Model:} Both firms know their own and their competitor's prediction accuracy, but not the correlation between the two prediction models. We characterize the uniquely individually rational and Pareto-optimal contracts for some notable cases. We also show that all contracts except inference-sharing can be optimal in this setting. 
 
 \item \textit{A Two Hypotheses Model:} One firm is able to determine the correct hypothesis during training, while the other has information about customers that is valuable during inference. 
 Here, we show that inference-sharing can be the unique individually rational and Pareto-optimal contract.  

\end{enumerate}
Overall, we conclude that each of the four train/inference combination contracts can be optimal:

\begin{itemize}

\item \textit{No-sharing} is the optimal individually-rational contract when the cost of making a wrong prediction is equal to the reward of making a correct prediction.
We show this first in Lemma~\ref{lem:KnownArbitrarySA}, for the case where the prediction model of each firm (based on its own data) is common knowledge, and then generalize it in Theorem~\ref{thm:no-sharing-symmetric-sa} for the general training-phase prior. This characterization follows from two main insights: (1) Under full-sharing, when the two firms share their inference-time signals, the firms will simply follow the primary firm signal. This is because a negative primary firm signal overshadows a positive secondary firm signal. (2) Given the first insight, the primary firm is only set to lose by sharing its signal, since the aggregate utility of the two firms is constant (and equals the accuracy of the primary firm's prediction). The secondary firm becomes more informed under full-sharing, and can extract the same utility as the primary firm.

\item \textit{Full-sharing} is the optimal individually-rational contract when the two firms have the same prediction accuracy. It is at least as good as no-sharing because the firms can use both signals to ``amplify'' or mitigate their individual signal (Theorem~\ref{thm:full-sharing-symmetric-prediction}), in a way that is mutually beneficial w.r.t. their equilibrium behavior under no-sharing. Full-sharing is also at least as good as train-sharing, because in both cases the equilibrium behavior is symmetric, but the full-sharing equilibrium is more informed (Lemma~\ref{lem:full_dominates_train_symmetric}). This is also true w.r.t. the infer-sharing equilibrium, even more generally (Lemma~\ref{lem:full_dominates_infer}), as the symmetry in the infer-sharing case stems not from having the same prediction accuracy, but from the fact that the infer-phase signals are shared, and the train-phase signals can not individually teach more about the correlation than the common prior. 

\item \textit{Train-sharing} is the optimal individually-rational contract when the two firms benefit from reaching different equilibria given a different correlation between their signals. In particular, the firms may prefer to each follow its signal when the correlation is low, but have the secondary firm 'yield' to the primary firm when the correlation is high and exit the market. Learning `when to quit' benefits the secondary firm as well, and so can emerge as the optimal individually-rational contract when full-sharing is `too permissive' for the primary firm to follow, due to loss in competitive advantage. 

\item \textit{Infer-sharing} is harder to come up with a situation where it is the optimal individually-rational contract. In fact, we show that in our correlation model it can not be the uniquely optimal individually-rational contract (Lemma~\ref{lem:full_dominates_infer}). In Section~\ref{sec:two_hyp} we explore a model we call ``the two hypotheses model'', which has a natural interpretation in health and scientific contexts, and show how infer-sharing may arise as the uniquely optimal individually-rational contract there (Theorem~\ref{thm:infer-sharing}). 

\end{itemize}

Beyond the existence results detailed above, which help provide intuition into the different types of prediction sharing contracts, the theorems of Section~\ref{sec:corrModel} also provide a partial characterization of our correlation model in several important cases such as symmetric utilities, or symmetric prediction accuracy. In Section~\ref{sec:simulation}, and further in Appendix~\ref{app:robustness}, we demonstrate how our abstract Bayesian model may be put into practice and implemented, using a real loan dataset. 

\subsection{Related Work}
Previous work in ML considered different aspects of strategic prediction. For example, \cite{Ben-PoratT17} and \cite{feng2022bias} study competition in prediction, based on shared and independent data, respectively. Literature on federated learning \cite[e.g.][]{CongYWY20, fraboni} considers free-riding by data providers to save costs while still benefiting from better predictions. 
\cite{GafniT22} and \cite{GradwohlT22} study data aggregation between competitors: In the former, a firm aims to exploit another firm's contributed data but also to mislead it, putting the integrity of the data-sharing protocol at risk. In the latter, segmentation information about consumers is split between firms, and firms decide whether or not to share their part of the data with others during the inference phase. The former work focuses solely on training models, while the second only on the inference stage of a known segmentation. 
Finally, some papers deal with the imbalance between firms with stronger and weaker models through the lens of fairness, leveraging tools from cooperative game theory \cite{modelSharingGames,ray2022fairness}.
There is also a growing economics literature on data markets \cite[see, e.g., the survey of][]{bergemann2019markets}
.
However, neither this economics literature, nor work on strategic ML consider strategic sharing of prediction models between competitors. They also do not contrast sharing during training and during inference, a distinction we see as crucial for ML in the distributed economy. 

There is some analogy between our work and the fundamental ML idea of aggregating weak learners into strong ones, and specifically to {\em bagging} and {\em stacking}. In bagging \cite{breiman1996bagging}, the ML algorithm deliberately creates subsets of the data and learns models for them in parallel; this is somewhat analogous to how, in our setting, different firms develop their own models. In stacking \cite{wolpert1992stacked}, there are two stages: first, models (derived, e.g., from bagging) produce predictions over a data-set, and second, a meta-learning algorithm learns how to generate a final authoritative prediction from the models' predictions. In a sense, our work can be viewed as {\em strategic} bagging and stacking.

\section{Model}
\label{sec:model}
\paragraph{Informational environment} There are two firms engaged in a competitive prediction task. Each firm obtains data in two phases: training and inference. In the training phase, examples with binary labels are drawn at random, and each firm learns a respective prediction model (i.e., classifier). The training phase may consist of one example, multiple examples, or ``infinitely many'' examples. In the inference phase firms use their learned model in order to predict the label of a new example. Firm 1's prediction is either $A$ or $B$ and firm 2's prediction is either $a$ or $b$, where the former indicates that the firm's prediction model believes the label is 1 and the latter indicates the label is 0. 
We model this interaction in an abstract Bayesian framework using the  {\em rich signal spaces} of \cite{green2022two} 
and \cite{gentzkow2017bayesian}.
We next describe the formal model, and then highlight the main elements and their interpretations.

A {\em world model $w$} consists of a prior distribution $\pi_w$ over $\{0,1\}$, as well as two {\em signal spaces}, one for each firm. For every true label $\tlabel\in\{0,1\}$, each signal space partitions $[0,1]$ into two sets, representing the probabilities associated with firms' prediction models, given true label $\tlabel$.\footnote{Formally, each signal space is a Lebesgue measurable bi-partition of $[0,1]\times\{0,1\}$.}  For the first firm, the first set is denoted $A_w^t\subset [0,1]$, and the second is denoted $B_w^t=[0,1]\setminus A_w^t$. For the second firm, the two sets are denoted $a_w^t$ and $b_w^t$.
Given $w$, a random example is modeled as a label $\tlabel$ drawn from $\{0,1\}$ according to $\pi_w$, as well as $\zeta$ drawn from $[0,1]$ uniformly at random.\footnote{The uniformity assumption here is without loss of generality.} Firm 1's signal (i.e., its model's suggested prediction under $w$) on this example is then $1$ if $\zeta \in A_w^t$ under $\tlabel$, and $0$ otherwise; firm 2's signal is $1$ if $\zeta\in a_w^t$ under $\tlabel$, and $0$ otherwise. 
In words, $\zeta$ chooses a ``location'' on the interval [0,1]. This location decides some signal for Firm 1 (according to the way it partitions the interval $[0,1]$), and similarly for Firm 2 (possibly with a different partition). Sampling $\zeta$ uniformly at random from $[0,1]$ is in a sense similar to sampling a random feature vector that is used to train the firms' prediction models / requires a prediction at inference time. 

In general, firms may not know the true $w$. Instead, let $W$ be a possibly infinite set of possible world models, and suppose there is a commonly known prior $\pi$ over them. An example of this framework is illustrated in Figure~\ref{fig:two_hyp_general_framework}.

Given this informational environment, the interaction proceeds as follows. In stage 0, Nature chooses an element $w$ of $W$ according to $\pi$. Then:
\begin{enumerate} 
    \item In the {\em training} stage, each firm $i$ obtains a {\em training signal} $w_i$ about the realized world model $w$. Each $w_i$ is a function of firm $i$'s respective signal space under $w$. Given signal $w_1$ (respectively, $w_2$) and the prior over $W$,  each firm $i$ uses Bayesian updating to derive posterior beliefs $\pi_i$ over world models $W$. 
    \item In the {\em inference} stage, $\zeta$ is drawn from $[0,1]$ uniformly at random, and a label $\tlabel$ is drawn from $\{0,1\}$ according to $\pi_w$. Firm 1  obtains the {\em inference-time signal} $X\in\{A,B\}$ that satisfies $\zeta\in X_{w'}^t$, where $w'\sim \pi_1$; firm 2  obtains the {\em inference-time signal} $x\in\{a,b\}$ that satisfies $\zeta\in x_{w'}^t$, where $w'\sim \pi_2$.\footnote{Notice that the inference-time signal is drawn according to the firm's \textit{posterior}, rather than according to some specific true possible world. This is since we are interested in calculating the firms' equilibrium behaviors, which follow their Bayesian perspective.}
    \item In the {\em action} stage, each firm $i$ takes an action $a_i\in\{0,1\}$. Utilities depend on both firms' actions, and true label $\tlabel$.
\end{enumerate}

Next, we consider different contracts for prediction sharing. Under {\em no-sharing}, the interaction proceeds as above. Under {\em train-sharing}, there is an additional stage between 1 and 2:
\begin{itemize}
    \item[1b.] Firms share their respective training signals $w_1$ and $w_2$.
\end{itemize}
Under {\em infer-sharing}, an additional stage between 2 and 3:
\begin{itemize}
    \item[2b.] Firms share respective inference-time signals $X$ and $x$.
\end{itemize}
Finally, under {\em full-sharing} both 1b and 2b take place.

\paragraph{Summary and interpretation} We now summarize the model elements and their interpretations:
\begin{itemize}
\item The {\em world model} $w$ is an information-theoretically optimal pair of classifiers for the firms.
    \item The {\em training signal} $w_i$ implies a posterior $\pi_i$ over world models, which we interpret as the actual classifier firm $i$ is able to train. We interpret the signal $w_i$ as firm $i$'s predictions on its labeled historical data. In practical terms, the training signal can be interpreted as the model that best fits the training data, out of all possible models. The firms can then share these signals (i.e., the functions or code representing their best models given their data), without sharing the data itself.

    \item The {\em inference-time signal} is the prediction ($X\in\{A,B\}$ for firm 1, $x\in\{a,b\}$ for firm 2) made by the trained classifier on an unlabeled inference-time example.
    \item Under {\em train-sharing}, firms share $w_1$ and $w_2$, their predictions on labeled historical data.
    \item Under {\em infer-sharing}, firms share $X$ and $x$, their respective predictions on the unlabeled inference-time example.
\end{itemize}

This formulation can capture a wide range of scenarios. The prior over $W$ implies a prior over the relative share $\pi_w$ of each label, a prior over the accuracy of each firm's model, and a prior over the correlation between the predictions of firms' models. 
The framework is illustrated in Figure~\ref{fig:two_hyp_general_framework}. See also Figure~1 and Figure~2 in \cite{gentzkow2017bayesian}. 

\begin{figure}[!htb]
\includegraphics[scale=0.4]{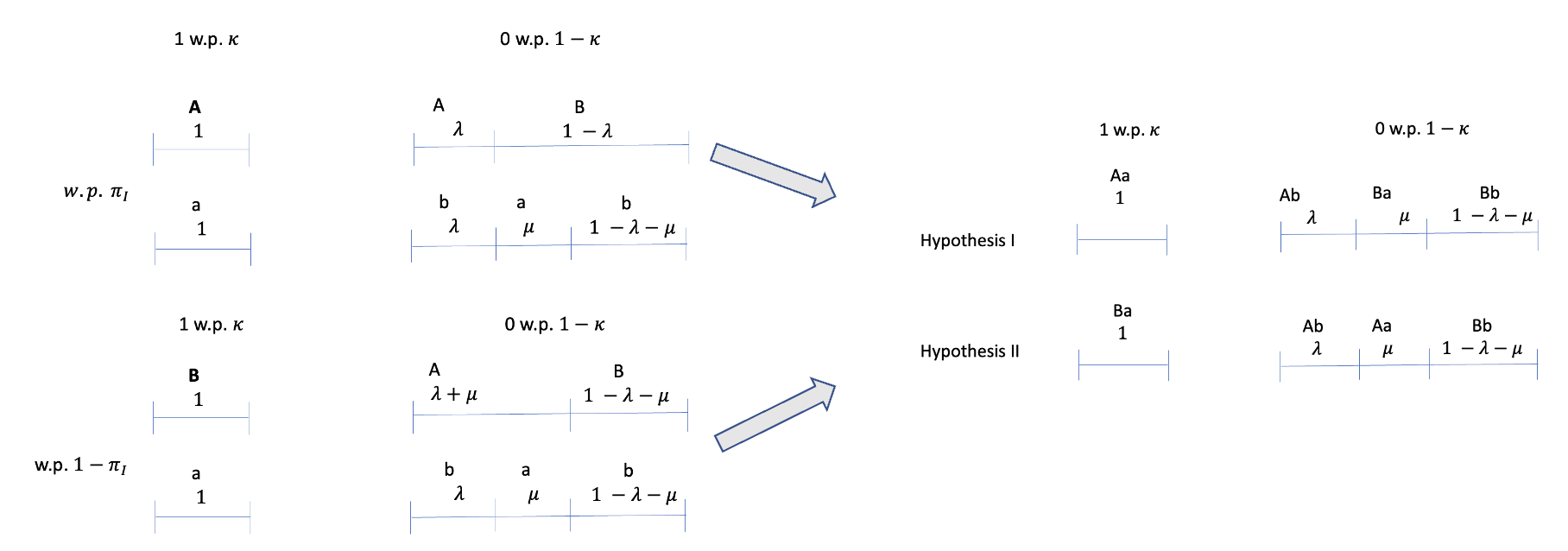}
\caption{There are two world models, represented by the top two and bottom two pairs of intervals, respectively. For both world models, $\pi_w=Pr[\tlabel=1]=\kappa$. In the first world, $A_w^1=[0,1]$ and $A_w^0=[0,\lambda]$. Thus, if $\tlabel=1$ firm 1 always obtains signal $A$, and if $\tlabel=0$  firm 1 obtains signal $A$ with probability $\lambda$---i.e., whenever $\zeta\in [0,\lambda]$---and signal $B$ with probability $1-\lambda$. Furthermore, $a_w^1=[0,1]$ and  $a_w^0=[\lambda,\lambda + \mu]$. Thus, if $\tlabel=1$ firm 2 always obtains signal $a$, and if $\tlabel=0$ obtains signal $a$ with probability $\mu$---i.e., whenever $\zeta\in [\lambda,\lambda + \mu]$---and signal $b$ with probability $1-\mu$. 
Finally, the bottom two pairs of line segments represent the firms' signal spaces in the second world model, which differs from the first only in firm 1's signal under $\tlabel=1$, namely, $A_w^1=\emptyset$ and $B_w^1=[0,1]$.
The interval structure of each of the firms results in a joint interval structure (and an induced joint probability over firm 1 signal $A/B$, firm 2 signal $a/b$, and the true realization $0/1$), shown on the rhs of the figure. In the infinite data model, where each of the firms learns its own interval structure with certainty, firm 1 is able to deduce the correct world model just by knowing its own interval structure. On the other hand, firm 2 does not learn (in a Bayesian sense) anything from its own interval structure. This example captures our ``Two Hypotheses'' model of Section~\ref{sec:two_hyp}.  
}
\label{fig:two_hyp_general_framework}
\end{figure}

\paragraph{Strategies} A  strategy $s_i$ of firm $i$ in the action stage is a mapping from the firm's signals to a distribution over actions $a_i\in\{0,1\}$. The firm's signals depend on the contract: under no-sharing, the respective signals are $\sigma_1^{ns}=(w_1, X)$ for firm 1 and $\sigma_2^{ns}=(w_2, x)$ for firm 2. Under train-sharing, they are $\sigma_1^{ts}=(w_1, w_2, X)$ and $\sigma_2^{ts}=(w_1, w_2, x)$. Under infer-sharing, they are $\sigma_1^{is}=(w_1, X, x)$ and $\sigma_2^{is}=(w_2, X, x)$. And under full sharing, both firms obtain signals $\sigma_i^{fs}=(w_1, w_2, X, x)$. 

\paragraph{Utility Model}
As noted above, utility $u_i(p,t,p')$ of  firm $i$ depends on 3 variables: The firm's action $p$, the true label $t$, and the other firm's action $p'$. For a given example, action $p$ is {\em correct} if it matches the example's label $t$. Given a training signal $w_i$, a contract $ct\in\{ns, ts, is, fs\}$, and a pair of strategies $(s_1,s_2)$
, the expected utility of firm $i$ is
\begin{equation}\label{eq:firms_utility}u_i^{ct}(w_i, s_1, s_2)=E\Big[u_i\Big(s_1\left(\sigma_1^{ct}\right),\tlabel, s_{2}\left(\sigma_{2}^{ct}\right)\Big)\Big],\end{equation}
where the expectation is over the draw of $w$ from $W$ according to $\pi|w_i$, the draw of $\tlabel$ according to $\pi_w$, the draw of $w_j$ under $w$, the draws of inference-time signals $X$ and $x$ under $w$, and the distributions of firms' randomization over actions. 

We  make some simplifying assumptions about utilities. First, we assume that 
\begin{equation}
\label{eq:symmetric_utils}
u_1 = u_2
\end{equation}. Second, we assume $u_i(p,t,p) = \frac{1}{2}u_i(p,t,\neg p)$, i.e., that if the two firms take the same action, the utility (whether positive or negative) is divided between them, in the sense that:
$$\sum_{i=1}^2 u_i(p,t,p) = \sum_{i=1}^2 \frac{1}{2}u_i(p,t,\neg p) \stackrel{\text{Eq.~\ref{eq:symmetric_utils}}}{=} u_1(p,t,\neg p).$$

To emphasize the notation, $u_i(p,t,p)$ is the utility when the other firm's prediction $p'$ is equal to $p$, and $u_i(p,t,\neg p)$ is the utility when the other firm's prediction $p'$ is \textit{different} than $p$. 

Thus, the ex-post utility is determined by four numbers: $R_0 = u_1(0,0,1), R_1 = u_1(1,1,0), C_0 = u_1(0,1,1), C_1 = u_1(1,0,0)$, where for example $R_0$ is the reward from correctly taking action $0$ while the other firm takes action $1$. We assume that $R_0, R_1 \geq 0$ and $C_0, C_1 \leq 0$. 

In the paper, we largely focus on two specific utility models that capture important settings. 
%
In Section~\ref{sec:corrModel}, we focus on a utility model we call {\em significant-action} utilities. In this model, there is a significant action---w.l.o.g., the action $1$. For example, this action may be choosing to issue a loan. When taking the other, safe action, both reward and cost satisfy $R_0 = C_0 = 0.$ 
If a firm takes a correct significant action exclusively, meaning that the other firm takes the safe action, it gets the full reward $R_1$. 
On the other hand, if a firm takes an incorrect significant action exclusively, it pays a cost $C_1$. 
When $C_1 = 1$, we call this the \textit{symmetric} significant-action utility model.

In Section~\ref{sec:two_hyp}, we focus on a utility model we call {\em matching recommendations} \cite[as in, e.g.,][]{CoopetitionAmazon}. In this model, there are no costs to a mistake---formally, $C_0 = C_1 = 0$---and there is a symmetric reward for any correct action---formally, $R_0 = R_1 = 1$. 
E.g., consider a firm that chooses between two possible recommendations to a user, and, if it correctly recommends what the user is looking for, the user will make a purchase. 

\paragraph{Equilibrium, individual rationality, and Pareto optimality}
Given training signals $w_1$ and $w_2$ 
and a contract $ct\in\{ns, ts, is, fs\}$, a pair of strategies $s=(s_i,s_{\neg i})$ form a {\em Nash equilibrium} at $(w_1, w_2)$ if for each $i$ and each strategy $s_i'$,
\begin{equation}u_i^{ct}(w_i, s_i, s_{\neg i}) \geq u_i^{ct}(w_i, s_i', s_{\neg i}).\end{equation}

An equivalent and perhaps more useful formulation of the equilibrium condition takes the perspective of the agent together with her beliefs \cite[see the discussion in Chapter 9 of ][]{maschler_solan_zamir_2013}. We define the utility from taking action $p_i \in \{0,1\}$ given the collection of signals $\sigma_i^{ct}$ 
and the other firm's strategy $s_{\neg i}$ as
\begin{equation}\label{eq:firm_cond_utility}\tilde{u}_i^{ct}(\sigma_i^{ct}, p_i, s_{\neg i})=E\Big[u_i\Big(p_i,\tlabel, s_{\neg i}\left(\sigma_{\neg i}^{ct}\right)\Big)\mid \sigma_i^{ct}\Big],\end{equation}
where the expectation is over the conditional draw of $\sigma_{\neg i}^{ct}$ and $\tlabel$ given signals $\sigma_i^{ct}$. We then say that $s$ is an equilibrium for firm $i$ if for every belief $\sigma_i^{ct}$ and every possible action $p_i' \in \{0,1\}$, 
\begin{equation}\label{eq:conditional_equilibrium_condition}\tilde{u}_i^{ct}(\sigma_i^{ct}, s_i(\sigma_i^{ct}), s_{\neg i})\geq \tilde{u}_i^{ct}(\sigma_i^{ct}, p_i', s_{\neg i}).\end{equation}


Next, a contract $ct$ {\em Pareto dominates} contract $ct'$ at $(w_1,w_2)$ if there exists an equilibrium $s$ under $ct$ such that, for every equilibrium $s'$ under $ct'$,
\begin{equation}u_1^{ct}(w_1, s) \geq u_1^{ct'}(w_1, s')~~~\mbox{ and }~~~u_2^{ct}(w_2, s) \geq u_2^{ct'}(w_2, s').\end{equation}
If at least one of the inequalities is strict then the Pareto dominance is {\em strict}. Contract $ct$ {\em Pareto dominates} contract $ct'$ if it Pareto dominates $ct'$ at every $(w_1, w_2)$, and in this case we write $ct \succeq ct'$. If $ct \succeq ct'$ and $ct' \succeq ct$, we write $ct = ct'$, and say that the two contracts are equivalent.  
Contract $ct$ {\em strictly Pareto dominates} $ct'$ if $ct \succeq ct'$ but $ct' \not\succeq ct$, and in this case we write $ct \succ ct'$.

Contract $ct$ is {\em individually rational (IR)} at $(w_1, w_2)$ either if it is the no-sharing contract (which we consider the default contract), or if $ct$ Pareto dominates the no-sharing contract at $(w_1,w_2)$. Contract $ct$ is {\em always IR} if it is IR at every $(w_1,w_2)$, namely, $ct \succeq ns$.

Contract $ct$ is {\em Pareto optimal} if it is not Pareto dominated by any other contract,  {\em Pareto-optimal IR} (IRPO) if it is both Pareto optimal and always IR, and  {\em uniquely IRPO} if it is the only contract that is both Pareto optimal and always IR.

We note that although our model is general, and can handle both mixed and pure Bayesian equilibria, our results in Section~\ref{sec:corrModel} onwards are for \textit{pure} Bayesian equilibria.

\section{Contracts for Prediction-Sharing with Unknown Correlation}
\label{sec:corrModel}

In this section we focus on the first of two specific settings within our framework. We assume firms have significant-action utilities and ``infinite data''. The latter assumption means that each firm's prediction model is in some sense an optimal classifier given the data features it is able to see. We believe that this is the most natural assumption to closely approximate massive data sets.\footnote{See also our analysis of a finite data case in  Section~\ref{subsec:finite_data}.} The main caveat is that neither firm knows the correlation between the firms' classifiers, even after learning its own classifier. We further assume that the prediction accuracy of each firm's classifier---formally, $Pr_{\pi_w}[1 | X = A]$ and $Pr_{\pi_w}[1 | x = a]$---are common knowledge.\footnote{We use $Pr[1]$ as shorthand for $Pr[t=1]$, and may omit the identifiers $\pi_w,X,x$ when clear from context.}
Finally, for simplicity we assume that $Pr[0] = Pr[1] = \frac{1}{2}$, and that the false-positive and false-negative rates  
 are the same for each of the firms:
\begin{equation}\alpha \stackrel{\mathrm{def}}{=} Pr[A | 1] = Pr[B | 0]~~~\mbox{ and }~~~ \beta \stackrel{\mathrm{def}}{=} Pr[a | 1] = Pr[b | 0].\end{equation} Still, the full \textbf{joint} distribution of the firms' pair of signals together with the true realizations under $w$ is unknown. As we see later in Section~\ref{subsec:unknownCorr}, this is equivalent to both firms not knowing how the signals of the two firms are correlated under the no-sharing contract, regardless of $(w_1, w_2)$. Firms also do not know the true label of the outcome they are trying to predict in the inference phase. We assume w.l.o.g.\ that $\alpha \geq \beta \geq \frac{1}{2}$.

At one extreme, it is possible that the firms' signals are independent. At the other extreme, it is possible that they are fully correlated. 
In 
Appendix~B we show how this model can be formulated using our general model from Section~\ref{sec:model}. 


\subsection{Warm-Up: Known Correlation}
\label{subsec:knownIndependent}

We start our investigation with a simple model in which the correlation between the firms is known. For an example in which firms' signals are known to be conditionally independent, see Figure~\ref{fig:known-corr}. 
\begin{figure}[!htb]
\includegraphics[scale=0.7]{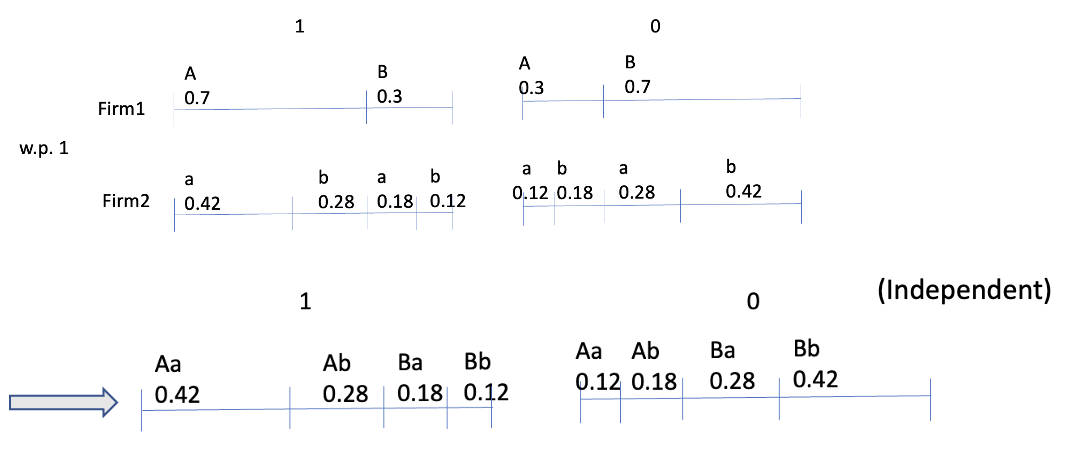}
\caption{Known correlation: An example of conditionally independent signals with $\alpha = 0.7, \beta = 0.6$. With one possible world, both firms know the joint distribution over true realizations and inference-time signals with certainty. 
}
\label{fig:known-corr}
\end{figure}

When the precision accuracy of both firms is common knowledge, as we assume throughout this section, then the \textit{correlation} between the firms' predictions fully determines the joint distribution of the pair of signals under label $t$. We show that formally in Claim~\ref{clm:cor-joint-dist}. By correlation we mean the Pearson correlation of the signals, namely 
$$\theta_t = \frac{Pr[X=A \land x = a | t] - \alpha \beta}{\sqrt{\alpha (1-\alpha) \beta (1 - \beta)}},$$

where $t$ is the true label realization. Notice that the two Bernoulli variables are the two firms' signals given the true realization. For simplicity, we assume that $\theta_1=\theta_0$, and denote the correlation simply by $\theta$.  

\begin{claim}\label{clm:cor-joint-dist}
In the correlation model, knowing the correlation $\theta$ determines the joint distribution of Firm1's signal $A/B$, Firm2's signal $a/b$, and the true realization $0/1$. 
\end{claim}

\begin{proof}
To see that the correlation determines the joint distribution in our setting, recall that for the Bernoulli variables in our settings, the Pearson correlation under label $1$ satisfies $\theta = \frac{Pr[X=A \land x = a | 1] - \alpha \beta}{\sqrt{\alpha (1-\alpha) \beta (1 - \beta)}}$. Thus, given $\theta$, $\alpha$, and $\beta$ we have 
\begin{equation}
\label{eq:Aa1_formula}
\begin{split}
& Pr[X=A \land x = a | 1] \\
& = \sqrt{\alpha \beta}\left(\sqrt{\alpha \beta} + \theta \cdot \sqrt{(1-\alpha)(1-\beta)}\right).
\end{split}
\end{equation}
This then determines $Pr[X=A \land x = b | 1] = \alpha - Pr[X=A \land x = a | 1], Pr[X=B \land x = a | 1] = \beta - Pr[X=A \land x = a | 1]$, and $Pr[X=B \land x = b | 1] = 1 - Pr[X = A \land x = a | 1] - Pr[X = A \land x = b | 1] - Pr[X = B \land x = a | 1]$. That is, it fully determines the joint distribution. 
For example, when $\alpha = \beta$ and $\theta = 0$ (i.e., the signals are conditionally independent), we have $Pr[X=A \land x = a | 1] = \alpha^2$, and when $\alpha = \beta$ and $\theta = 1$, we have $Pr[X=A \land x = a | 1] = \alpha$.
Finally, a symmetric argument holds under label 0.
\end{proof}



When the correlation is known, there is no added value in sharing $w_i$, since the world model $w$ is already known to both firms. Therefore, no-sharing is equivalent to train-sharing, and infer-sharing is equivalent to full-sharing. The only question is, which of these contracts, if any, is IRPO? 


\begin{restatable}[]{lemma}{knownSA}
    With known correlation and symmetric significant-action utilities,  only no-sharing and the equivalent train-sharing are IRPO. The unique equilibrium under these contracts has two regimes: A high $\beta$ regime where both firms play by their inference-time signals, and a low $\beta$ regime where Firm 2 ``gives in'' and always takes action $0$, while Firm 1 matches its action to its inference-time signal.
\label{lem:KnownArbitrarySA}
\end{restatable}

This matches what we learned to expect in practice: Firms develop their own classification models, and, assuming they are accurate enough, predict according to them. In Section~\ref{subsec:unknownCorr} we show, however, that once the correlation is not known with certainty, this conclusion may change, and full-sharing or train-sharing contracts may be uniquely IRPO. 


We also note that the threshold that separates the high and low $\beta$ regime is itself dependent on $\alpha$. The higher $\alpha$ is, the higher the threshold for the high beta regime, where Firm 2 follows its prediction signal. I.e., fixing Firm 2's prediction accuracy $\beta$, the firm is more likely to give in the higher Firm 1's prediction accuracy $\alpha$ is. 

Lemma~\ref{lem:KnownArbitrarySA} deals with symmetric significant-utilities. In the \textit{asymmetric} case, with a
higher cost for a mistake in the significant action $C_1$, but not so high as to prohibit ever taking a significant action altogether, the firms would prefer full-sharing, which enables them to take the significant action only when both receive positive signals. We show this in Appendix~\ref{sec:GeneralIndependent}.


\subsection{Unknown Correlation}
\label{subsec:unknownCorr}

 So far we have considered \textit{known} correlations. However, a more natural model is that the correlation is \textit{unknown}, and only some distribution over it is known. As we will see, this model can give rise to train-sharing as uniquely IRPO.

We begin with some preliminary lemmas. First, we show that within the specification of this subsection, full-sharing always Pareto dominates infer-sharing.
 \begin{restatable}[]{lemma}{fullDominatesInfer}
    \label{lem:full_dominates_infer}
 For any distribution $\pi_\theta$ over correlations and any $R_1$ and $C_1$, $fs \succeq is$. 
\end{restatable}

\begin{proof}
 
 In the correlation model, the private signal $w_i$ a firm gets during the training phase does not impact its posterior regarding the correlation $\theta$, which follows the distribution $\Theta$. Thus under infer-sharing, where each firm $i$ only sees $w_i$, we can ignore it, and we have $\sigma_1^{is} = \sigma_2^{is} = Xx$ for some pair of inference-time signal $X, x$. 
 
Thus, we can conclude that the infer-sharing equilibrium is symmetric between the firms. That is since as we argue above, the posterior for both firms after the training phase stays the same as the common prior. In the inference phase, both firms share their signals, and so both firms end up with the exact same information. Both firms' equilibrium strategy is to predict $1$ if and only if 
\[
\begin{split}& E_{\theta \sim \Theta}[Pr[1 | X = x_1, x = x_2, \theta] \\
& - C_1\cdot Pr[0 | X = x_1, x = x_2, \theta]] \geq 0.\end{split}\]
 
 Under full-sharing, a similar argument shows that for every pair of signals $Xx$ and correlation $\theta$ (which both firms learn during the training phase), the symmetric equilibrium strategy is to predict $1$ if and only if $Pr[1 | X = x_1, x = x_2, \theta] - C_1\cdot Pr[0 | X = x_1, x = x_2, \theta]$. 
 
 We can thus write, for the symmetric equilibrium strategies $s \stackrel{def}{=} s_1 = s_2$ of the infer-sharing contract,
 \[
 \begin{split}
     & u_i^{is} = E[u_i(s(X,x),t,s(X,x)] \\
     & = E_{\theta \sim \Theta} [E[u_i(s(X,x),t,s(X,x) | \theta]]  \\
 & \leq E_{\theta \sim \Theta} [E[\max_s u_i(s(X,x),t,s(X,x) | \theta]] = u_i^{fs}. 
 \end{split}
 \]
 \end{proof}

  Next, we see that train-sharing and no-sharing contracts are equivalent under sufficient symmetry.
 \begin{restatable}[]{lemma}{trainEqNoFirst}
    \label{lem:symmetric_reward_train_eq_no}
 If $R_1=C_1$ then $ts=ns$.
\end{restatable}

 \begin{proof}
 Suppose first that, under train-sharing, the firms follow the same equilibrium strategies $s_1, s_2$ for any realization $\theta \sim \Theta$.
 Then, it must be that, under no-sharing, $s_1, s_2$ is also an equilibrium: This is immediate since the IC conditions of Equation~\ref{eq:conditional_equilibrium_condition} under no-sharing follow immediately if the more granular IC conditions of the same equation under train-sharing are satisfied. 
 
 
 Now, we know by Lemma~\ref{lem:KnownArbitrarySA} that for any fixed $\theta$ the equilibrium strategies under train-sharing (which are the same as the equilibrium strategies for no-sharing given we know that the correlation is $\theta$) depend only on the values of $\alpha, \beta$, and so are independent of $\theta$. Thus, the same equilibrium strategies are played for any $\theta$. 

 \end{proof}

 

  \begin{restatable}[]{lemma}{fullDominatesTrainSymmetric}
    \label{lem:full_dominates_train_symmetric}
  If $\alpha = \beta$, then for any distribution $\pi_\theta$ over correlations and any $R_1$ and $C_1$, $fs \succeq ts$. 
\end{restatable}

 Finally, we can use the lemmas above to identify IRPO contracts. The two theorems below show that, under sufficient symmetry, only full-sharing or no-sharing are such contracts.
 
  \begin{restatable}[]{theorem}{NoSharingSymmetricSA}
    \label{thm:no-sharing-symmetric-sa}
If $R_1=C_1$ then no-sharing is uniquely IRPO.
\end{restatable}
 
 
   \begin{restatable}[]{theorem}{FullSharingSymmetricPrediction}
    \label{thm:full-sharing-symmetric-prediction}
 If $\alpha = \beta$ then full-sharing is either uniquely IRPO, or $fs=ns$ are the only IRPO.
\end{restatable}
 

However, outside the symmetries in Theorems~\ref{thm:no-sharing-symmetric-sa} and~\ref{thm:full-sharing-symmetric-prediction}, train-sharing can emerge as uniquely optimal.

\begin{restatable}[]{theorem}{TrainSharingExample}
    \label{thm:train-sharing-example}
 Train-sharing is uniquely IRPO for an open subset of parameters $\pi_\theta$, $\alpha$, $\beta$, $R_1$, $C_1$.
\end{restatable}

The intuition underlying the construction in the proof of Theorem~\ref{thm:train-sharing-example} is the following. Under no-sharing, the firms play the same equilibrium regardless of their train-phase signals $w_1$ and $w_2$. Under train-sharing, the equilibrium may depend on $w_1$ and $w_2$, and so in some cases may improve both firms' utilities relative to the no-sharing equilibrium. This happens particularly when the firms learn that their signals are highly correlated, which results in Firm 2 not taking a significant action (e.g., not issue a loan). This saves Firm 2 from attaining negative utility, and allows Firm 1 to fully exploit the utility of its predictions.

\section{A ``Two Hypotheses'' Model}
\label{sec:two_hyp}


In Section~\ref{sec:corrModel} we showed that all contracts except for infer-sharing can be uniquely optimal. In this section we complete the picture by describing a setting where infer-sharing is uniquely optimal. 
We focus on the second setting described in the introduction, which is summarized in Figure~\ref{fig:two_hyp_general_framework} of Section~\ref{sec:model}. 
We assume that firms have ``infinite data'' and matching-recommendations utilities: $C_0 = C_1 = 0$ and $R_0 = R_1 = 1$. 
 This setting captures a variety of natural circumstances, such as multi-factorial genetic disease and chemical testing. Consider a genetic disease that only manifests itself as a result of some environmental cause. There are two firms: Firm 1 performs genetic testing and knows (i) what genes cause the disease (ii) for a specific person, whether these genes are present. Firm 2, on the other hand, has users' behavioral data (e.g., credit card histories) and can identify the environmental cause. However, it does not understand the underlying genes that enable the disease. 

Formally, let $t=1$ denote the presence of the disease,  inference-time signals $A$ and $B$ denote the presence of two different gene mutations in the population, and inference-time signals $a$ and $b$ denote the presence and absence of the environment cause, respectively. There are two hypotheses: (I) the disease is caused by mutation $A$ and the environmental cause, and (II) the disease is caused by mutation $B$ and the environmental cause. Thus,  Hypothesis I (resp., Hypothesis II) is that firms see inference-time signals $Aa$ (resp., $Ba$) if and only if $t=1$. The following are common knowledge: 
\begin{itemize}
    \item Hypothesis I is correct with probability $\pi_I$, and Hypothesis II with probability $1-\pi_I$;
    
    \item without the environmental cause, the disease remains dormant ($Pr[0 | b] = 1$);
    
    \item the incidence rate of the disease in the general population is $Pr[1] = \kappa$;  and
    
    \item the incidence rates of the two different gene mutations in the general population are $Pr[A]=\kappa+(1-\kappa)\cdot \lambda$ and $Pr[B]=1-Pr[A]$.
\end{itemize}
 %
 
%
Our main result is that, within this setting, there are instances where infer-sharing is uniquely optimal.
\begin{restatable}[]{theorem}{InferSharingExample}
    \label{thm:infer-sharing}
    Infer-sharing is uniquely IRPO for an open subset of parameters $\pi_I$, $\kappa$, $\lambda$, and $\mu$.
\end{restatable}

The intuition underlying the construction in the proof of Theorem~\ref{thm:infer-sharing} is the following. Generally, in the two hypotheses model, Firm 1 has the ability to deduce the correct world model during training, even with only its own signal. In the cases we identify, train/full-sharing makes it lose this advantage, and thus can not be beneficial for it. We are left with no/infer-sharing as possible individually rational contracts. Since generally in the two hypotheses model, the signal of Firm 1 by itself is not enough to decide the user classification with certainty, infer-sharing helps in that Firm 1 can both determine the correct hypothesis and has the pair of signals that determines the true realization, and thus it always predicts correctly. In the cases we identify, the behavior of Firm 2 remains the same under both contracts, because of the fact that it can not deduce the correct world model during training. Hence, infer-sharing allows Firm 1 a ``free information meal'', similar to the example, given for a model that only captures inference stage sharing, without consideration of the training stage, in \cite{CoopetitionAmazon}.

\subsection{Beyond the Infinite-data Model}
\label{subsec:finite_data}
So far, we focused on the infinite-data model, where the training signal allows the firm to deduce the marginal distribution over its signal and the true realization. We conjecture that with enough data, the results are similar to the idealized infinite case that we analyze. However, with few samples, the results may change significantly. 
To demonstrate how the analysis may lead to different results when there is only little historical data, we consider the setting of Section~\ref{sec:two_hyp}, but when only one labeled example of past data is available to the firms. Thus, after the hypothesis (world) is drawn (Hypothesis $I$ w.p. $\pi_I$, and otherwise Hypothesis $II$), a sample is drawn from the joint distribution over the pair of signals and true realizations, and each firm sees its own signal and the true realization. I.e., if the true hypothesis is Hypothesis $I$, then Firm 1 sees $(A,0)$ w.p. $\alpha$, $(A,1)$ w.p. $\beta$, and $(B,1)$ w.p. $1 - \alpha - \beta$. The firms then update a Bayesian posterior over the world models. Under train-sharing and full-sharing, when historical predictions are shared, both firms see the entire sample, i.e., the pair of signals and the true realization. 

In the appendix, we prove that, in the two hypotheses model with the parameters used for Theorem~\ref{thm:infer-sharing} but with a single labeled example, the statement of Theorem~\ref{thm:infer-sharing} breaks down, as do some of the properties of equilibria derived in the theorem's proof. In particular:
\begin{restatable}[]{theorem}{oneSampleTwoHypotheses}
    \label{thm:oneSampleTwoHypotheses}
 Under the parameters of Theorem~\ref{thm:infer-sharing} but with one sample, no-sharing and train-sharing are not necessarily equivalent, no-sharing is IRPO (rather than infer-sharing), and Firm 1 has lower equilibrium expected utility than Firm 2. 
\end{restatable}


\section{Implementation for a Real Data-set}
\label{sec:simulation}

To see how our ideas may be put to practice, we use the peer-to-peer loan data of LendingClub, popularized by recent research such as \cite{emekter2015evaluating}, and publicly available at Kaggle \cite{lendingClubData}, to conduct a synthetic simulation. We take a random subset of 25\% of the features and assign it to Firm 1\footnote{In the Appendix, we include robustness tests where we vary the choice of features.}. We take another subset of 10\% of the features (possibly overlapping) and assign it to Firm 2. Vertically, we split the data into train, test and validation sets. We let each of the firms train a neural net over the training data (that includes only its features). Each neural net was trained for 20 epochs on a $8$-GB RAM M1 MacBook Pro, which takes about half an hour. The training signal consists of the neural net's predictions on whether loans are good or bad. The firms use the test data to learn the signal performance, which we assume then becomes common knowledge. Depending on the contract, the firms choose their equilibrium strategies based on the performance in the test data: under train-sharing and full-sharing they also see the other firm's predictions on the test data (rather than only knowing the aggregate performance measurements). The firms then use their models to get a signal for every example in the validation data. 
 Under infer-sharing and full-sharing they see the other firm's signals on the validation set, and may use it to alter their final actions. 
They are evaluated using their actions on the validation data, under significant action utilities with $R_1=1$ and cost $C_1$. 

Recall that our model assumes that firms share their training and inference-time signals, and not their entire data. 
 We thus compare the performance of \textit{full-sharing}---sharing of the firms' signals both on the historical data (Here: the test data) and the inference stage data (Here: the validation data)---with  \textit{total-sharing}---sharing of the firms' entire data, training a joint neural net model over the shared data, and dividing the utility that this model achieves on the validation data. 

There are several important aspects in which the practical implementation deviates from our formal model:
\begin{itemize}[leftmargin=*]
    \item We do not naturally have a common Bayesian prior over the joint distribution of signals and true realizations. We make the simplifying assumption that the signals are independent when relevant, i.e., under no-sharing and infer-sharing. Under train-sharing and full-sharing, we use the test data to learn the signals' correlations. 
    
    \item Simplifying symmetry assumptions---e.g., that the prediction accuracy is symmetric across labels---do not naturally appear in the real data-set, and so our decision rules need to adapt and do not exactly follow the ones given symmetry. 
    
    \item The calculation of equilibrium strategies over the test data leads to an empirical equilibrium that has some error when compared to a theoretical equilibrium taken in expectation. Moreover, the average utility of different contracts as calculated on the validation set may also have some error. 
    \end{itemize}

We find that the results generally follow the lines of our discussion in Section~\ref{sec:corrModel}: Varying by cost (going from $C_1 = 0$ to $C_1 = 2.5$ in $0.05$ steps), as summarized in Figure~\ref{fig:compare_contracts}, we find regimes where either full-sharing, no-sharing, or train-sharing are uniquely IRPO.
While full-sharing is almost always a Pareto optimal contract, there are significant regimes where it is not IR for firm 1, which results in the no-sharing and train-sharing regimes. In almost all instances and cost values of the simulation, infer-sharing is Pareto dominated by full-sharing, as predicted by Lemma~\ref{lem:full_dominates_infer}. 

The behavior of no-sharing and train-sharing is of particular interest. With low values of $C_1$, both contracts have the two firms issue a loan regardless of the signal. Then, with higher values of $C_1$, the firms move to an equilibrium where  each acts according to its signal, and later to an equilibrium where firm 1 predicts its signal while firm 2 does not issue any loans. 
At each such equilibrium shift, there is a discontinuity for firm 1's utility. For example, moving from each firm predicting its own signal to Firm 2 not issuing loans, allows it to get the full utility of its action instead of half. 

\begin{figure}[!htb]
\centering
\begin{minipage}{\textwidth}
\centering
\includegraphics[width=\linewidth]{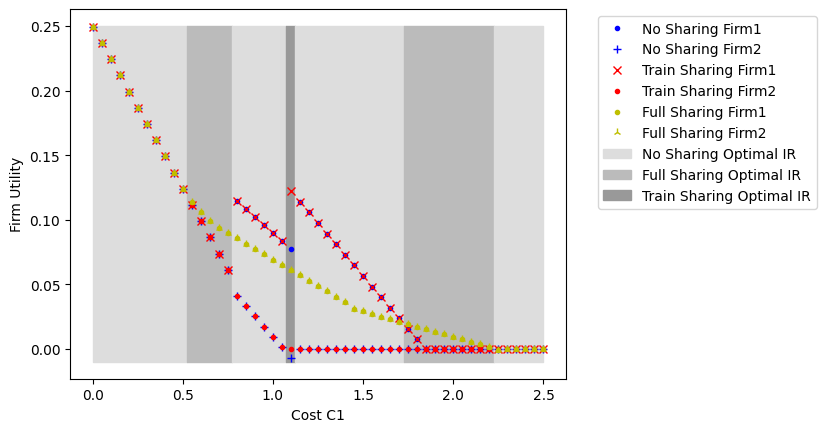}
\caption{No Sharing, Train Sharing and Full Sharing contracts performance for both firms and different costs. We do not include the infer sharing contract utility as they are very similar to (and dominated by) full sharing. We mark regimes where each contract is the optimal-welfare IR contract.}
\label{fig:compare_contracts}
\end{minipage}    \quad
\begin{minipage}{\textwidth}
\centering
\includegraphics[width=0.75\linewidth]{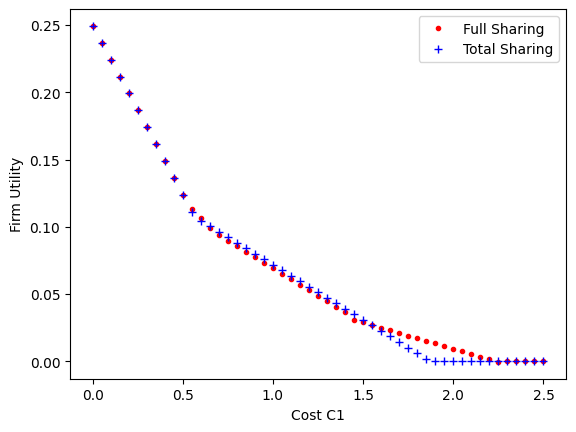}
\caption{Full Sharing vs. Total Sharing for different costs}
\label{fig:full_vs_total}
\end{minipage}
\end{figure}

In Figure~\ref{fig:full_vs_total}, 
we compare the performance of full-sharing with total-sharing for different costs. Importantly, both the full-sharing and total-sharing models are trained once for the symmetric cost $C_1 = 1$, and are then adapted to different costs by decision rules decided based on the test data. Somewhat surprisingly, they achieve very comparable performance, with even a slight advantage to full-sharing. We believe that this is due to the richer dual signal in the full-sharing case, which allows for more granular decision rules, mitigating the generally better prediction accuracy of total-sharing. 



\section{Discussion}\label{sec:discussion}

The analysis of incentives is a crucial aspect of the general effort to encourage data sharing, as recognized by the European Commission: ``In spite of the economic potential, data sharing between companies has not taken off at sufficient
scale. This is due to a lack of economic incentives (including the fear of losing a competitive edge)'' \cite{EC2020}. This paper introduces a novel element of data sharing---the distinction between sharing during training and inference---and demonstrates its importance to understanding firms' data-sharing incentives. 

Some natural questions arise as a result of our work: 

\begin{itemize}
\item We have assumed a common prior over priors for the firms. What if the firms have different beliefs? How robust is the emergence of uniquely optimal contracts to small differences in the epistemic models of the firms? 

    \item Our work is set within the framework of mechanism design without money, i.e., we suppose that firms share data based on mutual gain, rather than based on monetary compensation. In some cases it is natural to consider that one of the firms may compensate the other as part of the data sharing process. This could be interesting as future work and may build on the framework and insights we develop. 
\end{itemize}

\bibliographystyle{splncs04}
\bibliography{refer}

\begin{thebibliography}{10}
\providecommand{\url}[1]{\texttt{#1}}
\providecommand{\urlprefix}{URL }
\providecommand{\doi}[1]{https://doi.org/#1}

\bibitem{Ben-PoratT17}
Ben{-}Porat, O., Tennenholtz, M.: Best response regression. In: Advances in
  Neural Information Processing Systems 30: Annual Conference on Neural
  Information Processing Systems 2017, December 4-9, 2017, Long Beach, CA,
  {USA}. pp. 1499--1508 (2017)

\bibitem{bergemann2019markets}
Bergemann, D., Bonatti, A.: Markets for information: An introduction. Annual
  Review of Economics  \textbf{11},  85--107 (2019)

\bibitem{breiman1996bagging}
Breiman, L.: Bagging predictors. Machine learning  \textbf{24},  123--140
  (1996)

\bibitem{ray2022fairness}
Chaudhury, B.R., Li, L., Kang, M., Li, B., Mehta, R.: Fairness in federated
  learning via core-stability. Advances in Neural Information Processing
  Systems  \textbf{35},  5738--5750 (2022)

\bibitem{CongYWY20}
Cong, M., Yu, H., Weng, X., Yiu, S.: A game-theoretic framework for incentive
  mechanism design in federated learning. In: Yang, Q., Fan, L., Yu, H. (eds.)
  Federated Learning - Privacy and Incentive, Lecture Notes in Computer
  Science, vol. 12500, pp. 205--222. Springer (2020)

\bibitem{modelSharingGames}
Donahue, K., Kleinberg, J.: Model-sharing games: Analyzing federated learning
  under voluntary participation. Proceedings of the AAAI Conference on
  Artificial Intelligence  \textbf{35}(6),  5303--5311 (May 2021).
  \doi{10.1609/aaai.v35i6.16669},
  \url{https://ojs.aaai.org/index.php/AAAI/article/view/16669}

\bibitem{emekter2015evaluating}
Emekter, R., Tu, Y., Jirasakuldech, B., Lu, M.: Evaluating credit risk and loan
  performance in online peer-to-peer (p2p) lending. Applied Economics
  \textbf{47}(1),  54--70 (2015)

\bibitem{EC2020}
{European Commission}: A {European} strategy for data.
  \url{https://eur-lex.europa.eu/legal-content/EN/TXT/?qid=1593073685620&uri=CELEX:52020DC0066}
  (2020), accessed: 2021-05-13

\bibitem{feng2022bias}
Feng, Y., Gradwohl, R., Hartline, J., Johnsen, A., Nekipelov, D.: Bias-variance
  games. In: Proceedings of the 23rd ACM Conference on Economics and
  Computation. pp. 328--329 (2022)

\bibitem{fraboni}
Fraboni, Y., Vidal, R., Lorenzi, M.: {Free-rider Attacks on Model Aggregation
  in Federated Learning}. In: {AISTATS 2021 - 24th International Conference on
  Artificial Intelligence and Statistics} (2021)

\bibitem{FreundS97}
Freund, Y., Schapire, R.E.: A decision-theoretic generalization of on-line
  learning and an application to boosting. J. Comput. Syst. Sci.
  \textbf{55}(1),  119--139 (1997)

\bibitem{GafniT22}
Gafni, Y., Tennenholtz, M.: Long-term data sharing under exclusivity attacks.
  In: {EC} '22: The 23rd {ACM} Conference on Economics and Computation,
  Boulder, CO, USA, July 11 - 15, 2022. pp. 739--759. {ACM} (2022)

\bibitem{gentzkow2017bayesian}
Gentzkow, M., Kamenica, E.: Bayesian persuasion with multiple senders and rich
  signal spaces. Games and Economic Behavior  \textbf{104},  411--429 (2017)

\bibitem{GradwohlT22}
Gradwohl, R., Tennenholtz, M.: Pareto-improving data-sharing. In: FAccT '22:
  2022 {ACM} Conference on Fairness, Accountability, and Transparency, Seoul,
  Republic of Korea, June 21 - 24, 2022 (2022)

\bibitem{CoopetitionAmazon}
Gradwohl, R., Tennenholtz, M.: Coopetition against an amazon. Journal of
  Artificial Intelligence Research  \textbf{76},  1077--1116 (2023)

\bibitem{green2022two}
Green, J.R., Stokey, N.L.: Two representations of information structures and
  their comparisons. Decisions in Economics and Finance  \textbf{45}(2),
  541--547. {Originally circulated as IMSSS Technical Report No.\ 271, Stanford
  University, 1978} (2022)

\bibitem{maschler_solan_zamir_2013}
Maschler, M., Solan, E., Zamir, S.: Game Theory. Cambridge University Press
  (2013). \doi{10.1017/CBO9780511794216}

\bibitem{lendingClubData}
{Nathan George}: ``all lending club loan data''.
  \url{https://www.kaggle.com/datasets/wordsforthewise/lending-club} (2007),
  accessed: 2023-05-13

\bibitem{wolpert1992stacked}
Wolpert, D.H.: Stacked generalization. Neural networks  \textbf{5}(2),
  241--259 (1992)

\end{thebibliography}

\appendix

\section{Known Correlation Results and Proofs}
\label{sec:KnownArbitrarySA}
\label{sec:GeneralIndependent}

\knownSA*

\begin{proof}
We first emphasize again that when the prior over priors consists of a single possible prior, the first stage of learning is redundant, and so there is no difference between ``no-sharing'' and ``train-sharing'', and similarly no difference between ``infer-sharing'' and ``full-sharing''. Under no-sharing, a strategy $s^{no}_i$ of firm $i \in \{1,2\}$ is, given signal $x$, whether to predict $0$ or $1$. Under full-sharing, a strategy $s^{full}_i$ of firm $i$ is, given both firms' signals $x_1, x_2$, whether to predict $0$ or $1$. To specify a strategy (or part of a strategy), we sometimes use the notation $signal \rightarrow prediction$, e.g., under no sharing $A \rightarrow 0$ means that the primary firm predicts $0$ when it gets the signal $A$. 

Under \textit{full-sharing}, there is a unique symmetric equilibrium where $Aa \rightarrow 1, Ab \rightarrow 1, Ba \rightarrow 0, Bb \rightarrow 0$. I.e., both firms follow the primary firm's signal. This holds by the following argument. Under full-sharing, we have $\sigma_1^{fs} = \sigma_2^{fs}$. Moreover, in the known correlation case, $w_i$ can be ignored (as there is only one possible world model) and $\sigma_1^{fs}$ is of the form $Xx$ for some pair of inference-time signals $X \in \{A,B\}$ of Firm1 and $x \in \{a,b\}$ of Firm2. Let $p_i = s_i(Xx)$ be the prediction of firm $i$ if the pair of signals is $Xx$, and let $Q^{Xx}_{p_i,s_{\neg i}} = \begin{cases} 1 & p_i \neq s_{\neg i}(Xx) \\
\frac{1}{2} & p_i = s_{\neg i}(Xx)
\end{cases}$.

When $p_i = 1$, Equation~\ref{eq:firm_cond_utility} takes the form:

\begin{equation}
\begin{split}
      \tilde{u}_i^{fs}(\sigma_i^{fs}, 1, s_{\neg i}) &= E[u_i(1, t, s_{\neg i}(Xx))]\\ &= 
     Q^{Xx}_{1,s_{\neg i}} \cdot E[u_i(1, t, 0)] \\
     &= R_1 \cdot Pr[t = 1 | Xx] + C_1 \cdot Pr[t = 0 | Xx]\\  &=
     Pr[t = 1 | Xx] - Pr[t = 0 | Xx],
    \end{split}\end{equation}

where the first transition is since the signal that firm $\neg i$ sees is fixed to be the same one that firm $i$ sees, and so the expectation is only over the true realization $t$. The second transition is by the structure of our utility model. The third transition is by conditional expectation. The fourth transition is since in the symmetric significant action utility model $R_1 = 1, C_1 = -1$. 

When $p_i = 0$, Equation~\ref{eq:firm_cond_utility} takes the form:

\begin{equation}\begin{split}\tilde{u}_i^{fs}(\sigma_i^{fs}, 0, s_{\neg i}) &= E[u_i(0, t, s_{\neg i}(Xx))] \\ &= Q^{Xx}_{0,s_{\neg i}} \cdot E[u_i(0, t, 0)] \\ &= R_0 \cdot Pr[t = 1 | Xx] + C_0 \cdot Pr[t = 0 | Xx] \\ &= 0,\end{split}\end{equation}

where the last transition is since in the significant action utility model $R_0 = C_0 = 0$. We conclude that to satisfy the equilibrium condition of Equation~\ref{eq:conditional_equilibrium_condition}, it suffices to show that for every pair of signals $Xx$, both agents' strategies choose $s_i(Xx) = 1$ if and only if $Pr[t = 1 | Xx] - Pr[t = 0 | Xx] \geq 0$. 

Moreover, we know by Bayes' formula and our symmetry assumption ($Pr[0] = Pr[1]$) that:
$$Pr[1 | Xx] - Pr[0 | Xx] = \frac{\frac{1}{2}}{Pr[Xx]} \left( Pr[Xx | 1] - Pr[Xx | 0]\right),$$

and so the left-hand-side expression is non-negative iff the right-hand-side expression is non-negative. 

Let $\rho \stackrel{def}{=} Pr[X = A \land x = a | 1]$. Then we have
\[
\begin{split}
Pr[Aa | 1] - Pr[Aa | 0] = 
\rho - (1 - \alpha - \beta + \rho) \stackrel{\alpha \geq \beta \geq \frac{1}{2}}{\geq} 0, \\
\qquad Pr[ Ab | 1] - Pr[ Ab | 0] = \left(\alpha - \rho\right) - \left(\beta  - \rho\right) \stackrel{\alpha \geq \beta}{\geq} 0,  
\end{split}\]

and since $Pr[1 | Bb] - Pr[0 | Bb] = - \left(Pr[ 1 | Aa] - Pr[ 0 | Aa]\right), Pr[ 1 | Ba] - Pr[ 0 | Ba] = - \left(Pr[1 | Ab] - Pr[0 | Ab]\right),$ the inequalities are reversed for these signals. The utility for the primary firm is thus, following Eq.~\ref{eq:firms_utility},
\begin{equation}
\label{eq:u1_fs}
\begin{split}
 &u_1^{fs}(s_1,s_2) = 
 E_{X, x, t}[u_1(s_1(Xx),t,s_2(Xx))] \\ &= 
 \sum_{x_1 \in \{A,B\}} \sum_{x_2 \in \{a,b\}} \big( Pr[X = x_1 \land x = x_2]  \cdot Pr[1 | X = x_1 \land x = x_2] \cdot u_1(s_1(x_1, x_2), 1, s_2(x_1, x_2)) \big) \\
& + Pr[0 | X = x_1 \land x = x_2] u_1(s_1(x_1,x_2),0,s_2(x_1,x_2)) \big) \\ & \stackrel{\text{Bayes formula}}{=} 
 \sum_{x_1 \in \{A,B\}} \sum_{x_2 \in \{a,b\}} \big( Pr[1] \cdot Pr[X = x_1 \land x = x_2 | 1]  \cdot u_1(s_1(x_1,x_2), 1, s_2(x_1,x_2)) \\
 &+ Pr[0] \cdot Pr[X = x_1 \land x = x_2 | 0] \cdot u_1(s_1(x_1,x_2), 1, s_2(x_1,x_2))\big) ] \\
& = \frac{1}{2}\sum_{x_2 \in \{a, b\}} \big( Pr[1] \cdot Pr[X = A \land x = x_2 | 1]  - Pr[0] \cdot Pr[X = A \land x = x_2 | 0] \big) \\
& = \frac{1}{2} \left( Pr[1] \cdot Pr[X = A | 1] - Pr[0] \cdot Pr[X = A | 0] \right) \\
& = \frac{1}{4} \left(\alpha - (1 - \alpha)\right) = \frac{2\alpha - 1}{4}.
\end{split}
\end{equation}

Under \textit{no-sharing}, there are two regimes. Consider if the primary firm plays $A \rightarrow 1, B \rightarrow 0$. We argue that then, the secondary firm always plays $b \rightarrow 0$, since its utility from $b \rightarrow 1$ is (following Equation~\ref{eq:firm_cond_utility}):
\[
\begin{split}
     \tilde{u}_2^{ns}(b, 1, s_1) &= E[u_2(1, t, s_1(X)) | x = b] \\
    & =  \frac{1}{2} \cdot Pr[t = 1 \land X = A | x = b] + Pr[t = 1 \land X = B | x = b]\\
    & - \frac{1}{2} \cdot Pr[t = 0 \land X = A | x = b] - Pr[t = 0 \land X = B | x = b] \\
    & \stackrel{\text{Bayes Formula}}{=} \frac{1}{Pr[x = b]}\bigg( \frac{1}{2} Pr[x = b \land t = 1 \land X = A]  + Pr[x = b \land t = 1 \land X = B ] \\
    & - \frac{1}{2} \cdot Pr[x = b \land t = 0 \land X = A]  - Pr[ x = b \land t = 0 \land X = B] \bigg) \\ 
    & = 2 \bigg(\frac{\alpha - \rho}{4} + \frac{1 - \alpha - \beta + \rho}{2} - \frac{\rho}{2} - \frac{\beta - \rho}{4} \bigg) = \frac{1}{2}\left(2 - 2\rho - \alpha - 3\beta\right) \\
    & \stackrel{\alpha \geq \beta \geq \frac{1}{2}}{\leq} 0 = \tilde{u}_2^{ns}(b, 0, s_1)
\end{split}
\]

 As for the prediction given the signal $a$, the utility of Firm2 from $a\rightarrow 1$ is
 
 \[
\begin{split}
     \tilde{u}_2^{ns}(a, 1, s_1) & = E[u_2(1, t, s_1(X)) | x = a] \\
    & = \frac{1}{Pr[x = a]}\big( \frac{1}{2} Pr[x = a \land t = 1 \land X = A]  + Pr[x = a \land t = 1 \land X = B ] \\
    & - \frac{1}{2} \cdot Pr[x = a \land t = 0 \land X = A] - Pr[ x = a \land t = 0 \land X = B] \big) \\
    & = 2 \left(\frac{\rho}{4} + \frac{\beta - \rho}{2} - \frac{1 +\rho - \alpha - \beta}{4} - \frac{\alpha - \rho}{2} \right) \\
    & = \frac{1}{2} \left( 3\beta -1 -\alpha \right) 
\end{split}
\]
 
 Thus, the secondary firm best response is $a \rightarrow 1$ if and only if $3 \beta - \alpha - 1 \geq 1$. It is straightforward to verify that in both cases the strategies then form a unique equilibrium. 
 
 In the case that the secondary firm plays $a \rightarrow 0$, the equilibrium utility for the primary firm is (following Equation~\ref{eq:firms_utility}):
 
 \[
 \begin{split}
  u_1^{ns} & = E[u_1(s_1(X),t,0)] \\
 & = Pr[X = A \land t = 1] - Pr[X = A \land t = 0] \\
 & = \frac{1}{2} \left(\alpha - (1- \alpha)\right) = \frac{2\alpha - 1}{2} \geq \frac{2\alpha - 1}{4} \stackrel{\text{Eq.~\ref{eq:u1_fs}}}{=} u_1^{fs}.
 \end{split}
 \]
 
 In the case that the secondary firm plays $a \rightarrow 1$, 
the equilibrium utility for the primary firm is 
\[
\begin{split}
 u_1^{ns} & = E[u_1(s_1(X),t,s_2(x))] \\
& = \frac{1}{2}Pr[X = A \land t = 1 \land x = a] + Pr[X = A \land t = 1 \land x = b] \\
& - \frac{1}{2}Pr[X = A \land t = 0 \land x = a] - Pr[X = A \land t = 0 \land x = b] \\
& = \frac{1}{2}\left(\frac{\rho}{2} + (\alpha - \rho) - \frac{1 - \alpha - \beta + \rho}{2} - (\beta - \rho)\right)  \\
& = \frac{1}{4}\left(3\alpha - \beta - 1\right) = \frac{2\alpha - 1}{4} + \frac{1}{4}(\alpha - \beta) \geq \frac{2\alpha - 1}{4} \stackrel{\text{Eq.~\ref{eq:u1_fs}}}{=} u_1^{fs}.
\end{split}
\]
\end{proof}


\begin{lemma}
For any $\beta$ there is such $\alpha$ and a utility model (as defined by $C_0, C_1, R_0, R_1$) so that full-sharing is the unique feasible contract with known independent correlation. 
\end{lemma}

\begin{proof}
We give the idea of the construction. Consider asymmetric significant-action utilities, and fix (normalize) $R_1 = 1$.

Under full-sharing, as we know by the proof of Lemma~\ref{lem:KnownArbitrarySA}, the equilibrium strategies $s_1,s_2$ both have $s_i(Xx)$ for a pair of signals $Xx$ if and only if $Pr[Xx | 1] - C_1 \cdot Pr[Xx | 0] \geq 0$. The expression is monotone in the pair of signals $Xx$: It is highest for $Aa$, lower for $Ab$, even lower for $Ba$, and lowest for $Bb$. Thus, the equilibrium strategies are $Aa \rightarrow 1$ (i.e., the firms predict $1$ upon seeing $Aa$, and $0$ otherwise), if and only if:
\[
\begin{split}
    & Pr[Aa | 1 ] - C_1 Pr[Aa | 0] = \alpha \beta - C_1(1-\alpha)(1-\beta) \geq 0 \\
    & Pr[Ab | 1] - C_1 Pr[Ab | 0] = \alpha(1-\beta) - C_1\beta(1-\alpha) < 0,
    \end{split}
\]

where we use the expressions for independent correlation probabilities. 

Under no-sharing, we can follow the argument of Lemma~\ref{lem:KnownArbitrarySA} (but with a parametric cost $C_1$) to derive inequalities that guarantee an equilibrium where the firms predict according to their signals ($A \rightarrow 1, a \rightarrow 1$). Lastly, given that these are the no-sharing equilibrium strategies, assume that $u_2^{ns} \leq u_1^{ns} = \frac{\alpha\beta}{2} + \alpha(1-\beta) - C_1\left(\frac{(1-\alpha)(1-\beta)}{2} + \beta(1-\alpha)\right) \leq \frac{1}{2}\left(\alpha \beta - C_1(1-\alpha)(1-\beta)\right) = u_1^{fs} = u_2^{fs}$. If all these conditions hold, then full-sharing is the unique IRPO contract. This holds, for example, when $\alpha = 0.9, \beta = 0.85, C_1 = 2.5$. More generally, the conditions hold for any $\beta$ when $\alpha = \beta$ and $C_1 = \frac{2\beta^2 - 2\beta - 1}{2(\beta - 1)(\beta + 1)}$.   

\end{proof}

\section{Proofs for Unknown Correlation Case}
\label{app:unknown_corr}



  \fullDominatesTrainSymmetric*

  \begin{proof}

 Under train-sharing, the agents are symmetric, as they share their signals at the training phase, and have the same prediction accuracy at the inference phase. 
Since the agents are symmetric, we can assume a symmetric equilibrium is being played. Then, as we saw in the proof of Lemma~\ref{lem:full_dominates_infer}, this means the equilibrium strategy is to take a significant action (predict $1$) if and only if the expected utility is positive given the correlation learnt in the training phase (and without knowledge of the other firm's signal, as we are under train-sharing). Under full-sharing, the added information of the other firm's signal allows for more granular decisions when to take a significant action, and so $fs \succeq ts$. 
 \end{proof}
 
 \NoSharingSymmetricSA*
 
 \begin{proof}
 We know by Lemma~\ref{lem:full_dominates_infer} that $full-sharing \geq infer-sharing$ and by Lemma~\ref{lem:symmetric_reward_train_eq_no} that $train-sharing = no-sharing$. Thus, the only possibly IRPO contracts are $no-sharing$ or $full-sharing$. However, 
 in the symmetric case, we know from the proof of Lemma~\ref{lem:KnownArbitrarySA} that regardless of the correlation $\theta$, full-sharing always results in a symmetric equilibrium $s,s$ where $Aa \rightarrow 1, Ab \rightarrow 1$ (and $0$ otherwise). In essence, this is since these are the pairs of signals with probability $Pr[1 | Xx] > \frac{1}{2}$. 
 Under no-sharing, also by the proof of Lemma~\ref{lem:KnownArbitrarySA}, the equilibrium strategies $s_1, s_2$ do not depend on $\theta$, and $u_1^{ns} \geq u_1^{fs}$ for any fixed $\theta$. 
 Thus, the expected utility for the primary firm under full-sharing and any fixed correlation is $\frac{2\alpha - 1}{4}$, and since the expression does not depend on $\theta$, we have $u_1^{ns} = E_{\theta \sim \Theta} [E[u_1(s_1(X), t, s_2(x) | \theta]] \geq E_{\theta \sim \Theta} [E[u_1(s(X,x), t, s(X,x) | \theta]] = u_1^{fs}$. 
 
 \end{proof}
 
 \FullSharingSymmetricPrediction*
 
 \begin{proof}
 We consider both symmetric and asymmetric no-sharing equilibria. In a symmetric no-sharing equilibrium $s_1,s_2$, the symmetry means that $(s_1(A) = 1) \leftrightarrow (s_2(a) = 1), (s_1(B) = 1) \leftrightarrow (s_2(b) = 1)$. Under such a symmetric equilibrium, since $\alpha = \beta$, the utilities of the firms are the same. The utilities of the firms are the same for full-sharing as well, under the symmetric equilibrium $s,s$. 
 The sum of utilities of the firms satisfies

 \begin{align*}
     & u_1^{ns} + u_2^{ns} \\
     &=  E[u_1(s_1(X),t,s_2(x)) + u_2(s_1(X),t,s_2(x))]\\
     &= E_{\theta \sim \Theta} [E[u_1(s_1(X),t,s_2(x)) + u_2(s_1(X),t,s_2(x)) | \theta]]\\ 
     &= 
    E_{\theta \sim \Theta} [E[1[s_1(X) = 1 \lor s_2(x) = 1] u_1(1,t,0) | \theta]]\\ &= 
     E_{\theta \sim \Theta} [E[1[s_1(X) = 1 \lor s_2(x) = 1] (R_1 \cdot Pr[t = 1 | Xx] \\
     & - C_1 \cdot Pr[t = 0 | Xx]) | \theta]]\\ &\leq 
      E_{\theta \sim \Theta} [E[(R_1 \cdot Pr[t = 1 | Xx] - C_1 \cdot Pr[t = 0 | Xx])^+ | \theta]]\\ &= 
 E_{\theta \sim \Theta} [ E[1[s(X,x, \theta) = 1](R_1 \cdot Pr[t = 1 | Xx] \\
 & - C_1 \cdot Pr[t = 0 | Xx]) | \theta]]\\ &= 
   u_1^{fs} + u_2^{fs},
 \end{align*}
 
 where we use the notation $(x)^+ = \max \{x, 0\}$. We also use the fact that full-sharing takes a significant action if and only if $R_1 \cdot Pr[t = 1 | X, x, \theta] - C_1 \cdot Pr[t = 0 | Xx, \theta]$ holds, as we have seen in previous proofs. Therefore, and since $u_1^{fs} = u_2^{fs}$, each of the firms has at least as much expected utility under full-sharing than under no-sharing, which yields the theorem statement.
 
 As for asymmetric no-sharing equilibria, it can be directly calculated that the only possible such equilibrium when $\alpha = \beta$ is where $s_1$ is $A \rightarrow 1, B \rightarrow 0$, and $s_2$ always predicts $0$. Given this equilibrium, we know that 
 \begin{equation} \label{eq:asymmetric_equilibrium_equal_precision_a}
 \begin{split}
 & E_{\theta \sim \Theta}[R_1 \left( \frac{Pr[1 | Aa, \theta]}{2} + Pr[1 | Ba, \theta] \right)  - C_1 \cdot \left(\frac{[Pr[0 | Aa, \theta]}{2} + Pr[0 | Ba, \theta] \right)] < 0.
 \end{split}\end{equation}
 by the equilibrium condition that determines $a \rightarrow 1$. 
 
 Thus, 
 \[
 \begin{split}
  u_1^{ns} 
 & = E_{\theta \sim \Theta}[R_1 \left( Pr[1 | Aa, \theta] + Pr[1 | Ab, \theta]\right)  - C_1 \cdot \left([Pr[0 | Aa, \theta] + Pr[0 | Ab, \theta] \right)] \\
 & =  E_{\theta \sim \Theta}[R_1 \cdot \frac{ Pr[1 | Aa, \theta]}{2} - C_1 \cdot \frac{Pr[0 | Aa, \theta]}{2}] \\
 & + E_{\theta \sim \Theta}[R_1 \cdot \frac{Pr[1 | Aa, \theta]}{2} + Pr[1 | Ab, \theta] - C_1 \cdot \left(\frac{[Pr[0 | Aa, \theta]}{2} + Pr[0 | Ab, \theta] \right)] \\
 & \stackrel{\text{Eq.~\ref{eq:asymmetric_equilibrium_equal_precision_a}}}{<} E_{\theta}[R_1 \cdot  \frac{Pr[1 | Aa, \theta]}{2} - C_1 \cdot \frac{Pr[0 | Aa, \theta]}{2}] \\
 & \leq E_{\theta}[\left(R_1 \cdot \frac{Pr[1 | Aa, \theta]}{2} - C_1 \cdot \frac{Pr[0 | Aa, \theta]}{2} \right)^+] \\
 & \leq u_1^{fs}.
 \end{split}
 \]
 
 \end{proof}
 
 \TrainSharingExample*
 
 \begin{proof}
 
 We describe the general settings that yield our example, and then state such $\pi_{\theta}, \alpha, \beta, R_1, C_1$ values that implement it. Consider two possible worlds, with probability $w$ and $1-w$, respectively: In the first, signals are independent, and in the second, signals are totally correlated, i.e., $\theta$ is the maximal correlation possible between two firms with prediction accuracies $\alpha, \beta$. 
Notice that since \[
\begin{split}
& \beta = Pr[x=a | 1] \geq Pr[X=A \land x = a | 1]  \stackrel{Eq.~\ref{eq:Aa1_formula}}{=} \sqrt{\alpha \beta}\left(\sqrt{\alpha \beta} + \theta \cdot \sqrt{(1-\alpha)(1-\beta)}\right),
\end{split}
\]
we have $\theta \leq \sqrt{\frac{\beta(1-\alpha)}{\alpha(1-\beta)}}$, and $\theta$ is maximized when this inequality holds as an equality. When this happens, we have $Pr[X=A \land x = a | 1] = \beta, Pr[X=A \land x = b | 1] = \alpha - \beta, Pr[X=B \land x = a | 1] = 0, Pr[X=B \land x = b | 1] = 1 - \alpha$, and the mirror image of it conditional on $0$: $Pr[X=B \land x = B | 0] = \beta, Pr[X=B \land x = a | 0] = \alpha - \beta, Pr[X=A \land x = b | 0] = 0, Pr[X=A \land x = A | 0] = 1 - \alpha$. See Figure~\ref{fig:correct-modelling} for an illustration of the case when $\alpha = 0.7, \beta = 0.6$. 

For conciseness, in this proof we specify the strategies by what signals lead to predicting $1$, and all other signals lead to predicting $0$. 
Now, consider a case where the train-sharing equilibrium in the independent correlation case is $A \rightarrow 1, a \rightarrow 1$, and in the total correlation case is $A \rightarrow 1$, while under no-sharing the equilibrium is $A \rightarrow 1, a\rightarrow 1$.  With probability $w$, $\theta = 0$ and the equilibrium the firms play under both regimes is the same and so are the utilities. But w.p. $1-w$, the equilibrium played under the two regimes is not the same. This benefits the secondary firm by the fact that this is the train-sharing equilibrium: It means that given that $\theta = \theta_{max}$ and the primary firm plays $A\rightarrow 1$, the secondary firm's best response is to not take a significant action given any signal. But this also benefits the primary firm, since now it does not need to share its reward of $R_1 \cdot Pr[1 | Aa, \theta_{max}] - C_1 \cdot Pr[0 | Aa, \theta_{max}]$ with the secondary firm, and this expression is strictly positive as part of the equilibrium condition for the primary firm. Overall, this establishes ($train-sharing > no-sharing$). 

It is then enough, in order for $train-sharing$ to be the unique feasible contract, to require that it does not hold that $full-sharing > no-sharing$ (recall that by Lemma~\ref{lem:full_dominates_infer}, $full-sharing \geq infer-sharing$). This holds whenever the primary firm's utility is larger under no-sharing than under full-sharing. 

We can require that the full-sharing equilibrium both in the independent case and the totally correlated case is $Aa \rightarrow 1, Ab \rightarrow 1$, and this implies $$u_1^{fs} = \frac{1}{4} \left( w\left(R_1 \alpha - C_1(1-\alpha)\right) + (1-w)\left(R_1 \alpha - C_1(1- \beta)\right) \right).$$ 
We also have: 

\[
\begin{split}
  u_1^{ns} & = \frac{1}{2} \bigg( w\bigg(R_1 \left(\frac{\alpha \beta}{2} + \alpha(1-\beta)\right) - C_1\left( \frac{(1 - \alpha)(1 - \beta)}{2} + (1-\alpha) \beta \right)\bigg) \\
& + (1-w) \left(R_1 \left(\frac{\beta}{2} + \alpha - \beta\right) -C_1 \cdot \frac{1 - \beta}{2} \right)\bigg),
\end{split}
\]

and we impose the condition $u_1^{ns} > u_1^{fs}$. 

Finally, we note that the open set $0.72 < \alpha < 0.721, 0.513 < \beta < 0.514, 0.755 < C_1 < 0.756, 0.999 < R_1 < 1.001, 0.5 < w < 0.50001$ satisfies all the above conditions.
 
 \end{proof}
 
  \subsection{Representing the Correlation Model of Section~\ref{sec:corrModel} in our General Framework of Section~\ref{sec:model}}
  
First, we show how to describe our model of unknown correlations within the framework of our general model for firms' prediction-sharing. Consider the following construction. Fix some $\alpha, \beta$ and correlation distribution $\Theta$. We draw $\eta \sim \mathrm{Uniform}([0,1])$ and $\theta \sim \Theta$, and describe how the two parameters determine a world $w_{\eta, \theta}$. In a world $w_{\eta, \theta}$, the interval structure of firm 1 satisfies 
\[
\begin{split} & A_w = \\
& \begin{cases}
\{ [\eta, \eta + \alpha] \times \{1\} \cup ([0,1]\setminus [\eta, \eta + \alpha]) \times \{0\} \} \\
\qquad \qquad \qquad \qquad \eta + \alpha \leq 1 \\
\{ [\eta, 1] \cup [0, \eta + \alpha - 1] \times \{1\} \\
\cup ([0,1]\setminus ([\eta, \eta + \alpha] \cup [0, \eta + \alpha - 1) )) \times \{0\} \} \\
\qquad \qquad \qquad \qquad \eta + \alpha > 1 \\
\end{cases}.\end{split}\] We let $\rho = \sqrt{\alpha \beta}(\sqrt{\alpha \beta} + \theta \cdot \sqrt{(1-\alpha)(1-\beta)}$ as in Equation~\ref{eq:Aa1_formula}. 
Let 
\[
\begin{split}
& e_w = \\
& \begin{cases}
[\eta + \alpha - \rho, \eta + \alpha - \rho + \beta] \\ \qquad \qquad \qquad \qquad \eta + \alpha - \rho + \beta \leq 1 \\
[\eta + \alpha - \rho, 1] \cup [0, \beta - (1 - (\eta + \alpha - \rho))] \\
\qquad \qquad \qquad \qquad  \eta + \alpha - \rho + \beta > 1 \land \eta + \alpha - \rho \leq 1 \\
[\eta + \alpha - \rho - 1,\beta + \eta + \alpha - \rho - 1] \\
\qquad \qquad \qquad \qquad \eta + \alpha - \rho > 1 \land \eta + \alpha - \rho + \beta \leq 2 \\
[\eta + \alpha - \rho - 1,1] \cup [0, \beta - (1 - (\eta + \alpha - \rho - 1))] \\
\qquad \qquad \qquad \qquad  \eta + \alpha - \rho > 1 \land \eta + \alpha - \rho + \beta > 2\\
\end{cases}.
\end{split}
\]
The interval structure of firm 2 then satisfies
$a_w = \{e_w \times \{1\} \cup ([0,1]\setminus e_w) \times \{0\} \}$. 

The main technical claim for the construction is that not only firm 1, but also firm 2, gets a uniform draw over its intervals $a_w$. That is, if $\mathrm{frac}(x) = x - \lfloor x \rfloor$ is the fractional part of a number $x$, then $\mathrm{frac}(\eta + \alpha - \rho) \sim \mathrm{Uniform}([0,1])$ (up to a measure zero adjustment at $0$ and $1$). 
To see this, observe that, for every $\rho$, $\mathrm{frac}(\eta + \alpha - \rho) | \rho \sim \mathrm{Uniform}([0,1])$, and so the probability density function $\tilde{p}(\mathrm{frac}(\eta + \alpha + \rho)) = E_{\rho \sim \rho(\Theta)} [\mathrm{frac}(\eta + \alpha + \rho) | \rho] = 1$.

We can then verify that the firms' posterior over the correlation between the firms' signals is $\Theta$ regardless of the training signal $w_i$---their own interval structure---that they observe. As we show in Section~\ref{subsec:knownIndependent}, knowing $Pr[X = A \land x = a | 1]$ uniquely determines $\theta$, and vice-versa. In our construction, any world has $Pr[X = A \land x = a | 1] = \rho$, which is a 1-to-1 function of $\theta$. So, firm 1, upon learning $A_w$, knows that the distribution over $\rho$ is determined by $\Theta$, and so concludes that the distribution over correlations is $\Theta$. Similarly, firm 2, upon learning $a_w$ (which is determined by $\eta + \alpha - \rho$), knows that $\eta$ is drawn uniformly over $[0,1]$, and so by Bayes' rule $\tilde{p}(\rho | \eta + \alpha - \rho) = \tilde{p}(\eta + \alpha - \rho | \rho) \frac{\tilde{p}(\rho)}{\tilde{p}(\eta + \alpha - \rho)} = \tilde{p}(\rho)$, where the last transition follows from the firms' uniform prior over world models. 

To further illustrate this construction, let us use the example we repeatedly use in our proofs, where the correlation (conditioned on the true realization) between the firms' signal is either $0$ (conditional independence) or the maximal possible (total correlation). Figure~\ref{fig:incorrect-modelling} shows an intuitive but \textbf{incorrect} way to model this case within our general framework. Figure~\ref{fig:correct-modelling} shows a simple correct modelling using a finite number of worlds (which is \textit{not} our general construction). Figure~\ref{fig:general_construction} shows how we model this case using our general construction. 

\begin{figure}[!htb]
\includegraphics[scale=0.6]{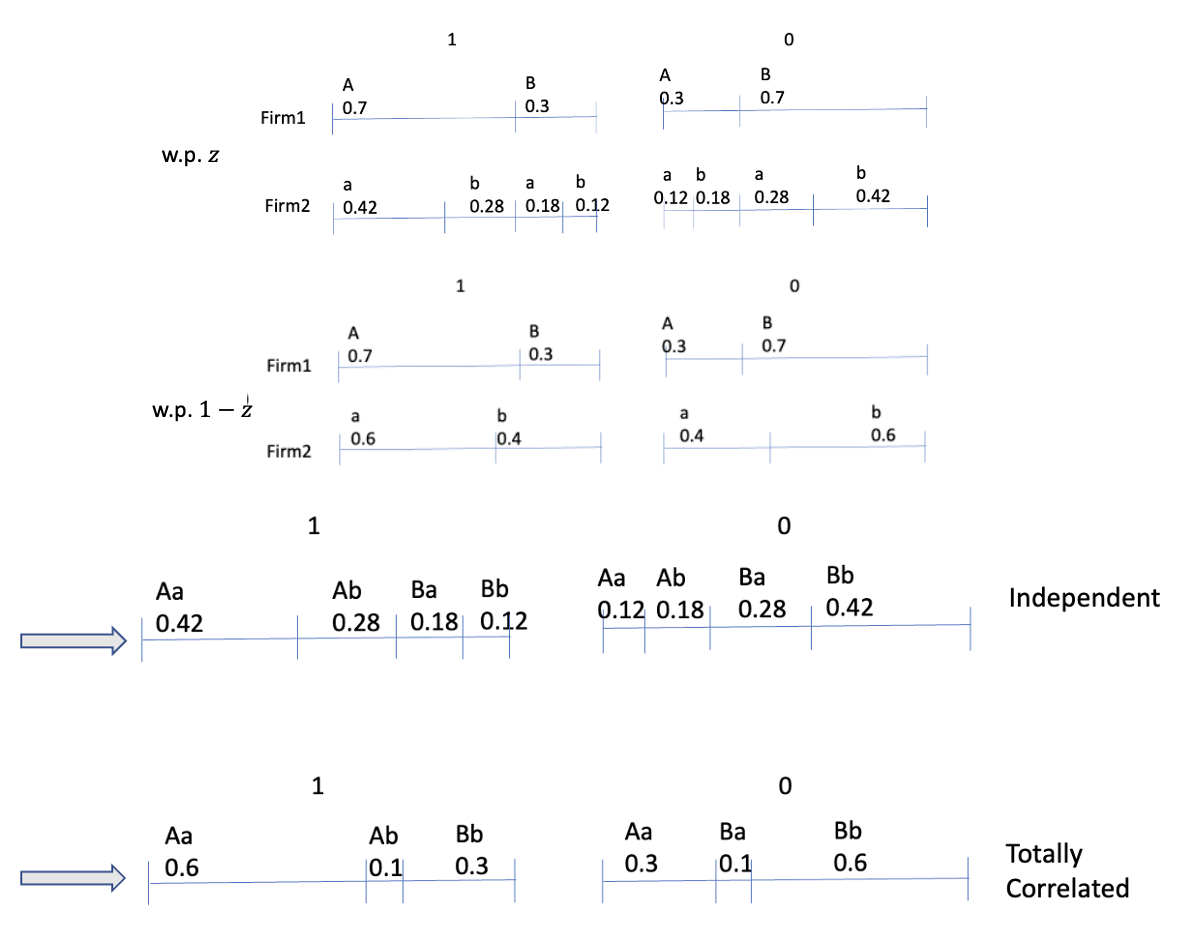}
\caption{\textbf{Incorrect} Modelling of the unknown correlation within the general prediction-sharing framework. In this modelling, we have $\alpha = 0.7, \beta = 0.6$, and there are two possible worlds, one appears w.p. $z$, and results in independent signals $A, a$ and $B,b$ given the true realization (whether $0$ or $1$). The other appears w.p. $1-z$ and results in totally correlated signals. However, this modelling does not capture our model of unknown correlation, since firm 2 can deduce the correlation based only on its own information.
}
\label{fig:incorrect-modelling}
\end{figure}

\begin{figure}[!htb]
\includegraphics[scale=0.8]{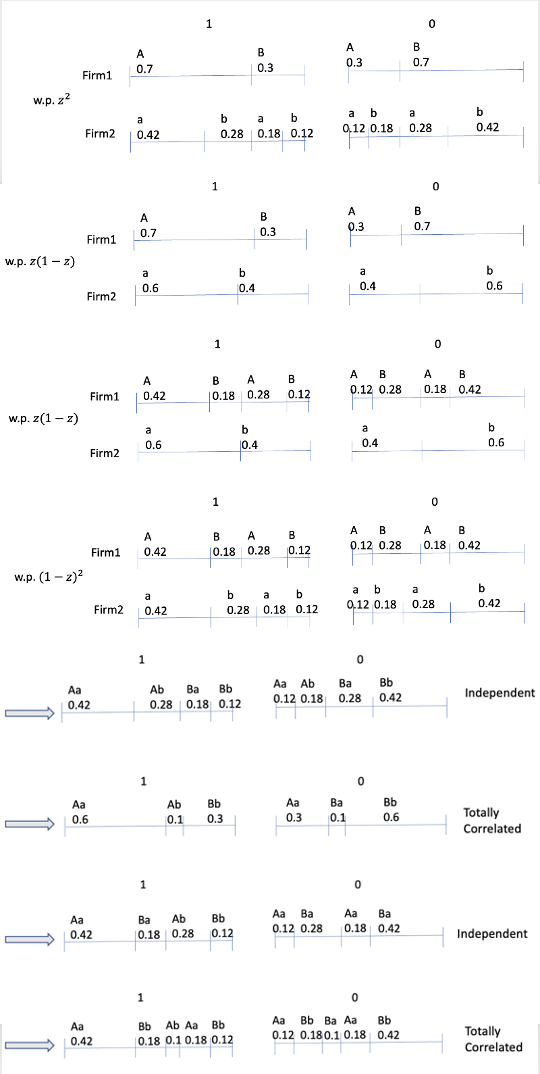}
\caption{\textbf{Correct} Modelling of the unknown correlation within the general prediction-sharing framework. In this modelling, we have $\alpha = 0.7, \beta = 0.6$, and there are \textit{four} possible worlds, appearing respectively w.p. $z^2, (1-z)z, (1-z)z,$ and $(1-z)^2$. The first and third possible worlds result in independent signals $A, a$ and $B,b$ given the true realization (whether $0$ or $1$). The second and fourth possible worlds result in totally correlated signals. 
In this modelling, whatever interval structure firm 1 sees, a Bayesian updating of the posterior would lead it to believe that the correlation between the firms' signal is independent w.p. $z$ and totally correlated w.p. $1-z$. The same holds for firm 2. 
}
\label{fig:correct-modelling}
\end{figure}

\begin{figure}[!htb]
\includegraphics[scale=0.8]{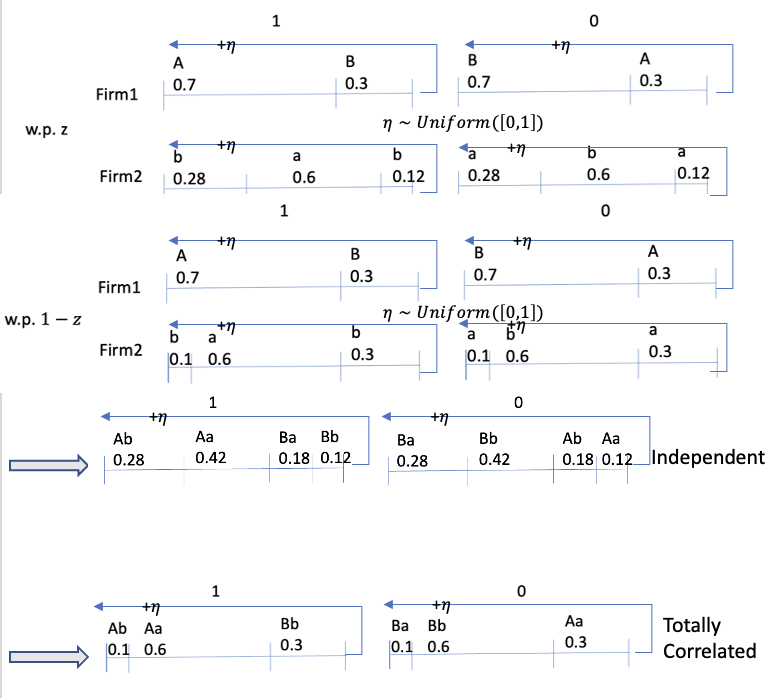}
\caption{\textbf{General \& Correct} Modelling of the unknown correlation within the general prediction-sharing framework, using our \textit{general construction}. In this modelling, we have $\alpha = 0.7, \beta = 0.6$, and there are \textit{infinite} possible worlds, drawn either from the upper type (representing the independent signals case) w.p. $z$ or the lower type (representing the totally correlated signals case) w.p. $1-z$, and then the intervals as described in the figure are shifted cyclically with an offset $\eta \sim \mathrm{Uniform}([0,1])$. 
In this modelling, whatever interval structure firm 1 sees, a Bayesian updating of the posterior would lead it to believe that the correlation between the firms' signal is independent w.p. $z$ and totally correlated w.p. $1-z$, and the same holds for firm 2. 
}
\label{fig:general_construction}
\end{figure}

\clearpage
 
 \section{Two Hypotheses Model Proofs}

\begin{figure}[!htb]
\includegraphics[scale=0.4]{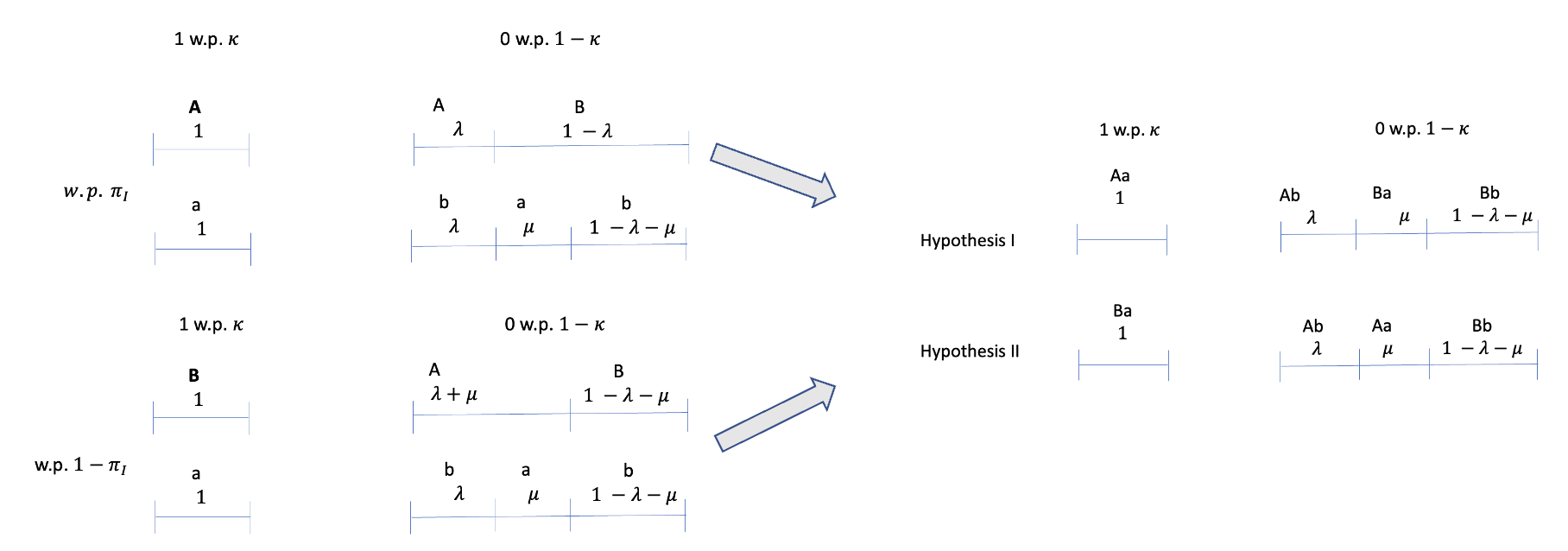}
\caption{For the reader's convenience we include Figure~\ref{fig:two_hyp_general_framework} again here. 
}
\label{fig:two_hyp_general_framework2}
\end{figure}
 
 \InferSharingExample*
 \begin{proof}
We describe conditions that yield the stated result.


(Conditions that yield $no-sharing = train-sharing$)

Consider 

\begin{equation}
\label{eq:thm4_ns=ts}
(1-\kappa) \mu > 2\kappa, 
\frac{1}{2}(1 - \kappa)(1 - \lambda - \mu)
> \kappa > \frac{(1-\kappa)\lambda}{2}.\end{equation} Then, 

Since $w_1$ (the signal that Firm1 receives regarding the true model) is enough to determine the correct hypothesis with certainty, under train-sharing both firms know the correct hypothesis given $\sigma_i^{ts}$. Let $\sigma_2^{ts} = \{w_1, w_2, a\}$, i.e., Firm2 sees $w_1, w_2$ in the training phase and the signal $a$ in the inference phase. For any strategy $s_1$ of Firm1, and train-phase signals $w_1, w_2$, we have


\[
\begin{split}
& \tilde{u}_2^{ts}(\{w_1, w_2, a\}, 0, s_1) =E\Big[u_2\Big(0,\tlabel, s_1(\{w_1, w_2, X\})\Big)\mid  \{w_1, w_2, a\}\Big] \\ &\geq 
 E\Big[u_2\Big(0,\tlabel, 0\Big)\mid  \{w_1, w_2, a\}\Big] =  \frac{1}{2} Pr[\tlabel = 0 | \{w_1,w_2, a\} ]\\ &= 
 \frac{1}{2}Pr[\tlabel = 0 | a ] = \frac{1}{2}\frac{(1-\kappa)\mu}{\kappa + (1-\kappa)\mu} \\ & \stackrel{\text{Eq.~\ref{eq:thm4_ns=ts}}}{>} 
 \frac{\kappa}{\kappa + (1-\kappa)\mu} = Pr[\tlabel = 1 | \{w_1,w_2, a\} ] = 
 E\Big[u_2\Big(1,\tlabel, 0\Big)\mid  \{w_1, w_2, a\}\Big]\\ &\geq 
 E\Big[u_2\Big(0,\tlabel, s_1(\{w_1, w_2, X\})\Big)\mid  \{w_1, w_2, a\}\Big] \\
 & = \tilde{u}_2^{ts}(\{w_1, w_2, a\}, 1, s_1),
\end{split}
\]

i.e., it is dominant for secondary firm to play $a \rightarrow 0$ regardless of the hypothesis or the primary's firm action. Since the signal $b$ always coincides with true realization $0$, the secondary firm always plays $b \rightarrow 0$ as well. But since predicting $0$ is the dominant strategy for any training-phase signal under train-sharing, as we saw in previous proofs, this is also the dominant strategy under $no-sharing$ for the secondary firm. 

As for the primary firm, because it has a different signal for the true realization $1$ under each hypothesis, it can determine which hypothesis is true both in the $no-sharing$ and $train-sharing$ contracts. Overall, this means that the equilibrium for both contracts is that the secondary firm always predicts $0$, and the primary firm, given that it finds that hypothesis I is correct, predicts $A \rightarrow 1, B \rightarrow 0$, and given that it finds that hypothesis II is correct, always predicts $0$ (by a direct calculation of the conditional utilities for Firm1 and using Equation~\ref{eq:thm4_ns=ts}). 

(Condition that prevents $full-sharing$ being always individually rational)

Consider 
\begin{equation}
\label{eq:thm4_fs_not_IR}
\kappa > (1-\kappa)\lambda.\end{equation} 

Under full-sharing, both firms can determine the true hypothesis and given the pair of signals can also follow up with predicting the correct true realization with certainty. Thus, the expected utility under any $w_1, w_2$, for each firm, is $\frac{1}{2}$. To show that it does not hold that full-sharing is always individually rational, it thus suffices to show that $u_1^{ns} > \frac{1}{2}$ with some $(w_1,w_2)$. There are two possible $w_1$ for the primary firm: Either that the true realization $1$ always coincides with $A$ (which happens if and only if Hypothesis I is true), or that it always coincides with $B$. If the first holds, and by our assumption in Equation~\ref{eq:thm4_fs_not_IR}, for the equilibrium strategies we saw in our analysis of no-sharing, 
$ u_1^{ns} = \kappa + \frac{(1-\kappa)(1-\lambda)}{2} = \frac{2\kappa + 1 - \kappa - \lambda + \kappa\lambda}{2} = \frac{1 + \kappa - (1-\kappa)\lambda}{2} > \frac{1}{2} = u_1^{fs}$. 

(Conditions for $infer-sharing > no-sharing$)

We note that under infer-sharing, the primary firm can both determine the correct hypothesis and has the pair of signals that determines the true realization, and thus always predicts correctly. We wish to find a condition so that the secondary firm has the same strategy as under $no-sharing$, to always predict $0$. Under both hypotheses, the secondary firm should predict $Ab, Bb \rightarrow 0$. As for signal $Aa$, the secondary firm predicts $0$, as long as
$(1-\pi_I) (1-\kappa) \mu > \pi_I \kappa$. Similarly for signal $Ba$, the secondary firm predicts $0$ as long as $\pi_I (1-\kappa)\mu > (1-\pi_I) \kappa$. When this holds, the primary firm gains utility through its ability to differentiate between the $Ba$ and $Bb$ signals under hypothesis II. This ``free meal'' phenomenon is similar to the example in the introduction of \cite{CoopetitionAmazon}. 

We wrap up by noting that all the above conditions are satisfied when $0.2 < \mu < 0.38, 0.05 < \kappa < 0.06, 0.011 < \lambda < 0.02$, and  $0.65 < \pi_I < 0.75$. 
\end{proof}

\oneSampleTwoHypotheses*


\begin{proof}
We consider additional conditions, on top of these of Theorem 4, that would yield the stated result: An example where firm 2 has higher expected utility in the IRPO contract than firm 1, and that all contracts besides no-sharing are not individually rational. 

(Under train-sharing)
In the two hypotheses model with one sample, the pair of signals $Aa$ and $Ba$, together with their true realization, determines with certainty the true hypothesis. On the other hand, the signals $Ab$ or $Bb$ together with their true realization (that can only be $0$) maintains the same posterior as the prior. Thus, we conclude that under the train-sharing and full-sharing contracts, there is a probability of $\kappa + (1-\kappa)\mu$ that both firms learn the true hypothesis, and otherwise the firms maintain their prior. 

The secondary firm always predicts $0$ (the dominance argument of Theorem~\ref{thm:infer-sharing} generalizes regardless of the training-phase signal), and the primary firm predicts $A \rightarrow 1, B \rightarrow 0$ if it knows Hypothesis I is correct, and always predicts $0$ if it knows Hypothesis II is correct. Otherwise, its posterior is the same as the prior, and so its prediction for $A$ is $0$ if and only if $\frac{(1-\kappa)(\lambda + (1-\pi_I)\mu)}{2} \geq \pi_I \kappa $, and its prediction for $B$ is $0$ if and only if $\frac{(1-\kappa)(1 - \lambda - (1-\pi_I)\mu)}{2} \geq (1-\pi_I) \kappa$. For the parameters used in Theorem~\ref{thm:infer-sharing}, this results in $A \rightarrow 1, B \rightarrow 0$. The expected utility of the primary firm is thus 
\[
\begin{split}
& Pr[Hypothesis I] \cdot \left( Pr[A \land 1| Hypothesis I] + \frac{1}{2}Pr[B \land 0]\right) \\
& + Pr[Hypothesis II]\bigg( Pr[s_1 = A\rightarrow 0, B\rightarrow 0 | Hypothesis II]  \cdot \frac{1}{2}Pr[0 | Hypothesis II] \\
& + Pr[s_1 = A\rightarrow 1, B\rightarrow 0 | Hypothesis II]  \cdot \left(\frac{1}{2}Pr[B \land 0 | Hypothesis II] \right) \bigg) \\
& = \pi_I \cdot (\kappa + \frac{(1 - \kappa)(1 - \lambda)}{2})  + (1-\pi_I) ((\kappa + (1-\kappa)\mu) \cdot \frac{1 - \kappa}{2}  + (1 - \kappa)(1-\mu) \cdot \frac{(1-\kappa)(1 - \lambda - \mu)}{2}), \end{split}
\]
and the utility of the secondary firm is 

\[
\begin{split}
& Pr[Hypothesis I] \cdot \left( Pr[A \land 0| Hypothesis I] + \frac{1}{2}Pr[B \land 0]\right) \\
& + Pr[Hypothesis II] \bigg( Pr[s_1 = A\rightarrow 0, B\rightarrow 0 | Hypothesis II] \cdot\frac{1}{2}Pr[0 | Hypothesis II] \\
& + Pr[s_1 = A\rightarrow 1, B\rightarrow 0 | Hypothesis II] \cdot \left(Pr[A \land 0 | Hypothesis II] \cdot \frac{1}{2}Pr[B \land 0 | Hypothesis II] \right) \bigg) \\
& = \pi_I \cdot ((1-\kappa)\lambda + \frac{(1 - \kappa)(1 - \lambda)}{2})  + (1-\pi_I) ((\kappa + (1-\kappa)\mu) \cdot \frac{1 - \kappa}{2} \\
& + (1 - \kappa)(1-\mu) \cdot \frac{(1-\kappa)(\lambda + \mu) + (1-\kappa)(1 - \lambda - \mu)}{2}). \end{split}
\]


(Under no-sharing)

The secondary firm always predicts $0$. Thus, the strategy of the primary firm follows a direct application of Equation~\ref{eq:conditional_equilibrium_condition} where we have the strategy $s_2 = 0$ of firm 2.

The primary firm knows Hypothesis I is correct with certainty when it sees $(A,1)$ (which happens w.p. $\pi_I \kappa$), and then predicts $A \rightarrow 1, B \rightarrow 0$. It knows Hypothesis II is correct with certainty when it sees $(B,1)$ (which happens w.p. $(1-\pi_I) \kappa$) and then always predicts $0$. For $(A,0)$, we have $Pr[(A,0) | Hypothesis I] = (1-\kappa)\lambda, Pr[(A,0) | Hypothesis II] = (1-\kappa)(\lambda + \mu)$, and so $w' = Pr[Hypothesis I | (A,0)] = Pr[(A,0) | Hypothesis I] \frac{Pr[Hypothesis I]}{Pr[(A,0)]} = (1-\kappa)\lambda \cdot \frac{\pi_I}{\pi_I (1- \kappa)\lambda + (1-\pi_I)(1-\kappa)(\lambda + \mu)}$. Similarly, for $(B,0)$, we have $w'' = Pr[Hypothesis I | (B,0)] = \frac{(1 -\kappa)(1 - \lambda)\pi_I}{(1 - \kappa)(1-\lambda)\pi_I + (1-\pi_I)(1-\kappa)(1 - \lambda - \mu)}$. 
We conclude (similarly to as we did in the train-sharing case for the original prior) that if the primary firm sees $(A,0)$ it always predicts $0$, and if the primary firm sees $(B,0)$, it predicts $A \rightarrow 1, B \rightarrow 0$. The primary firm's expected utility is then 
\[
\begin{split}
& Pr[Hypothesis I] \cdot \bigg(Pr[(B,0) \lor (A,1) | Hypothesis I]  \cdot \bigg(Pr[A \land 1| Hypothesis I] + \frac{1}{2}Pr[B \land 0]\bigg) \\
& + Pr[(A,0) | Hypothesis I] \cdot \left(\frac{1}{2}Pr[0 | Hypothesis I]\right)\bigg) \\
& + Pr[Hypothesis II] \cdot \bigg( Pr[(B,1) \lor (A,0) | Hypothesis II] \cdot \frac{1}{2}Pr[0 | Hypothesis II] \\
& + Pr[(B,0) | Hypothesis II] \cdot \left(\frac{1}{2}Pr[B \land 0 | Hypothesis II] \right) \bigg) \\
& = \pi_I \cdot \bigg((1 - (1-\kappa)\lambda) (\kappa + \frac{(1-\kappa)(1-\lambda)}{2}) + (1-\kappa)\lambda \frac{1-\kappa}{2}\bigg) \\
& + (1-\pi_I) \cdot \bigg((1 - (1-\kappa)(1-\lambda - \mu)) \frac{1-\kappa}{2} \\
& + (1-\kappa)(1 - \lambda - \mu) \left(\frac{(1-\kappa)(1-\lambda-\mu)}{2} \right),
\end{split}
\]

and firm2's expected utility is
\[
\begin{split}
& Pr[Hypothesis I] \cdot \bigg(Pr[(B,0) \lor (A,1) | Hypothesis I]  \cdot \bigg(Pr[A \land 0| Hypothesis I] + \frac{1}{2}Pr[B \land 0]\bigg) \\
& + Pr[(A,0) | Hypothesis I] \cdot \left(\frac{1}{2}Pr[0 | Hypothesis I]\right)\bigg) \\
& + Pr[Hypothesis II] \cdot \bigg( Pr[(B,1) \lor (A,0) | Hypothesis II] \cdot \frac{1}{2}Pr[0 | Hypothesis II] \\
& + Pr[(B,0) | Hypothesis II] \cdot \bigg(Pr[A \land 0 | Hypothesis II] + \frac{1}{2}Pr[B \land 0 | Hypothesis II] \bigg) \bigg) \\
& = \pi_I \cdot \bigg((1 - (1-\kappa)\lambda) ((1-\kappa)\lambda + \frac{(1-\kappa)(1-\lambda)}{2}) + (1-\kappa)\lambda \frac{1-\kappa}{2}\bigg) \\
& + (1-\pi_I) \cdot \bigg((1 - (1-\kappa)(1-\lambda - \mu)) \frac{1-\kappa}{2} \\
& + (1-\kappa)(1 - \lambda - \mu) \left((1-\kappa)(\lambda + \mu) + \frac{(1-\kappa)(1-\lambda-\mu)}{2} \right).
\end{split}
\]

(Under infer-sharing)

The Bayesian updating phase based on the historical sample is the same as in the no-sharing case, and so the primary firm attains the various posteriors under the same probabilities. The secondary firm, regardless on the sample, has the same posterior as the prior, as the two hypotheses look the same for its signal structure. For the pair of signals $Ab$ and $Bb$, both firms always predict $0$. For the signal $Ba$, since $\frac{\pi_I(1-\kappa)\mu}{2} > (1-\pi_I) \kappa$, it is dominant for the secondary firm to predict $0$ regardless of how the primary firm predicts, and similarly for $Aa$, since $\frac{(1-\pi_I)(1-\kappa)\mu}{2} > \pi_I \kappa$, it is also dominant for the secondary firm to predict $0$. 


For the signal $Ba$, since the secondary firm always predicts $0$, the primary firm predicts $1$ for it when its posterior is $w'$,
or when it knows with certainty that Hypothesis II is correct. It predicts $0$ for it when its posterior is $w''$, since $$w''\frac{(1-\kappa)\mu}{2} > (1-w'') \kappa,$$
and also when it knows with certainty that Hypothesis I is correct.

For the signal $Aa$, since the secondary firm always predicts $0$, the primary firm predicts $0$ for it when its posterior is $w'$,
or when it knows with certainty that Hypothesis II is correct. It predicts $1$ for it when its posterior is $w''$, since $$w''\kappa > (1-w'') \frac{(1-\kappa)\mu}{2},$$
and also when it knows with certainty that Hypothesis I is correct.

The primary firm's utility is then:

\[
\begin{split}
   & Pr[Hypothesis I] \bigg( Pr[(B,0) \lor (A,1) | Hypothesis I] 
   \cdot ( Pr[1| Hypothesis I] + \frac{1}{2}Pr[0 | Hypothesis I] ) \\
   & + Pr[(A,0) | Hypothesis I](\frac{1}{2} Pr[0 \land b | Hypothesis I]) \bigg) \\
   & + Pr[Hypothesis II] \bigg( Pr[(B,0) | Hypothesis II] \cdot ( \frac{1}{2} Pr[0 \land b | Hypothesis II]) \\
   & + Pr[(A,0) \lor (B,1) | Hypothesis II](\frac{1}{2} Pr[0 | Hypothesis II] + Pr[1 | Hypothesis II]) \bigg) \\
   & = \pi_I \bigg( (1 - (1-\kappa)\lambda)(\kappa + \frac{1-\kappa}{2}) + (1-\kappa)\lambda \frac{(1-\kappa)(1-\mu)}{2}\bigg) \\
   & + (1-\pi_I) \bigg((1-\kappa)(1-\lambda - \mu) \frac{(1-\kappa)(1-\mu)}{2} \\
   & + (1 - (1-\kappa)(1-\lambda - \mu))(\frac{1-\kappa}{2} + \kappa)  \bigg).
\end{split}
\]

The utility of the secondary firm is:

\[
\begin{split}
   & Pr[Hypothesis I] \bigg( Pr[(B,0) \lor (A,1) | Hypothesis I] \cdot (\frac{1}{2}Pr[0 | Hypothesis I] ) \\
   & + Pr[(A,0) | Hypothesis I](\frac{1}{2} Pr[0 \land b | Hypothesis I] + Pr[0 \land a | Hypothesis I]) \bigg) \\
   & + Pr[Hypothesis II] \bigg( Pr[(B,0) | Hypothesis II] \cdot ( \frac{1}{2} Pr[0 \land b | Hypothesis II] \\
   & + Pr[0 \land a | Hypothesis II]) + Pr[(A,0) \lor (B,1) | Hypothesis II] \cdot (\frac{1}{2} Pr[0 | Hypothesis II] ) \bigg) \\
   & = \pi_I \bigg( (1 - (1-\kappa)\lambda)( \frac{1-\kappa}{2}) + (1-\kappa)\lambda (\frac{(1-\kappa)(1-\mu)}{2} + (1-\kappa)\mu)\bigg) \\
   & + (1-\pi_I) \bigg((1-\kappa)(1-\lambda - \mu) (\frac{(1-\kappa)(1-\mu)}{2} + (1-\kappa)\mu) \\
   & + (1 - (1-\kappa)(1-\lambda - \mu))(\frac{1-\kappa}{2})  \bigg).
\end{split}
\]



(Under full-sharing)
The Bayesian updating phase based on the historical sample is the same as in the train-sharing case, and so both firms (that see the full single sample together) either know the hypothesis with certainty, or maintain the prior $\pi_I$. 

Given that the firms maintain the prior $\pi_I$, upon seeing the signals $Ab$ or $Bb$ they both predict $0$. If they see the signal $Ba$, since
$$\pi_I \frac{(1-\kappa)\mu}{2} > (1-\pi_I)\kappa,$$
both firms predict $0$. If they see the signal $Aa$, since
$$(1-\pi_I) \frac{(1-\kappa)\mu}{2} > \pi_I\kappa,$$ 
both firms predict $0$. 

We thus have that the utility of firm 2 is, if we denote 
$$p^* = Pr[(Aa,1) \lor (Ba,1) \lor (Aa,0) \lor (Ba,0)] = \kappa + (1-\kappa)\mu, $$
is $$\frac{1}{2} (p^* + (1-p^*) (1-\kappa)) = \frac{1}{2}(\kappa (\kappa + (1-\kappa)\mu - 1) + 1).$$

Finally, we note that the parameter choice $\kappa = \frac{5}{32}, \lambda = \frac{1}{8}, \mu = \frac{1}{2}, \pi_I = \frac{1}{2}$ satisfies both the conditions detailed above and the conditions detailed in the proof of Theorem~\ref{thm:infer-sharing}. 

\end{proof}

\section{Robustness of Our Results in Section~\ref{sec:simulation}}
\label{app:robustness}

We wish to deepen our empirical results of Section~\ref{sec:simulation}. In particular, we wish to add more variability into the feature selection process. For this purpose, we introduce the following sampling method. First, for every $\epsilon \in \{1, 0.85, 0.7, 0.55, 0.4, 0.25\}$, we randomly sample $\epsilon$ of the data-set features. Denote the resulting partial set of features $X$. Next, we randomly sample $\frac{0.25}{\epsilon}$ of the features in $X$ to decide the features for Firm1. Similarly, we sample $\frac{0.1}{\epsilon}$ of the features in $X$ to decide the features for Firm2. The end result is, as in Section~\ref{sec:simulation}, that Firm1 sees $25\%$ random features out of the original, and Firm2 sees $10\%$. However, one could hope that the features (and so, the resulting models after training) will be more correlated as $\epsilon$ is smaller, since an overlap of the features the firms see is more likely. For each value of $\epsilon$, we run $32$ random experiments, where we repeat the above process, followed by the analysis described in Section~\ref{sec:simulation}. Overall this results in $32 \cdot 7 = 224$ experiments where we train a random model for each of the firms. 

We then quantify the emergence of optimal contracts in the following way. For each cost in the range $0.5$ to $1.5$ in $0.05$ steps, we find what are the optimal contracts (in the IRPO sense). We break equivalencies in favor of the more ``natural'' contracts, i.e., if any contract has exactly the same expected utilities as no-sharing, we would not consider it optimal, and if train-sharing or infer-sharing are equivalent to full-sharing, we would similarly not consider them optimal. If by the end of this process we have more than one optimal contract for a specific cost, we divide the `benefit' between all the optimal contracts. Overall, for each experiment, we get a score for each contract of the frequency it is optimal. Since we are interested in seeing the possible influence of correlation on optimality, we directly calculate the correlation of Firm1's model and Firm2's model predictions over the validation set, using the Matthews correlation for a confusion matrix. In this way, we can directly compare the two variables we are interested in, instead of using the indirect $\epsilon$ parameter we use as part of our process, which we only expect to have a probabilistic negative connection with the correlation (i.e., as $\epsilon$ higher, we could expect lower correlation). 

We present our results in Figure~\ref{fig:freq-vs-corr}. We also provide a smoothed presentation of our results in Figure~\ref{fig:smoothed-freq-vs-corr}, where for every $0.1$ range of the correlation that appeared in the experiments (i.e., $-0.1$ to $0$, $0$ to $0.1$, and so on, up to $0.8$ to $0.9$), we average the correlation that appear in the range, and the frequencies associated with them. 

\begin{figure}[!htb]
\centering
\includegraphics[scale=0.5]{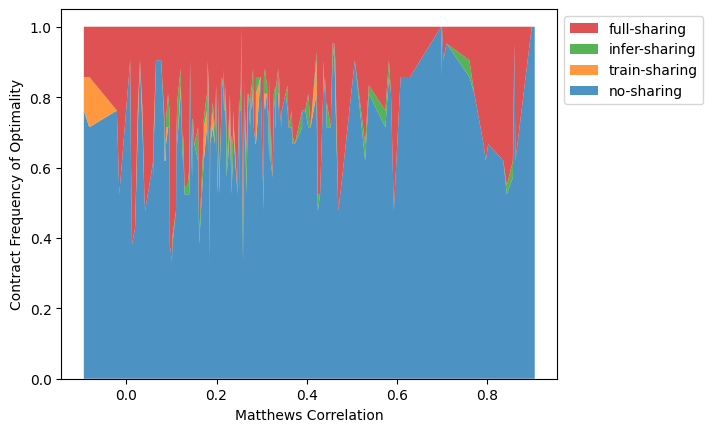}
\caption{Contract Optimality vs. Firm Models' Correlation
}
\label{fig:freq-vs-corr}
\end{figure}

\begin{figure}[!htb]
\centering
\includegraphics[scale=0.5]{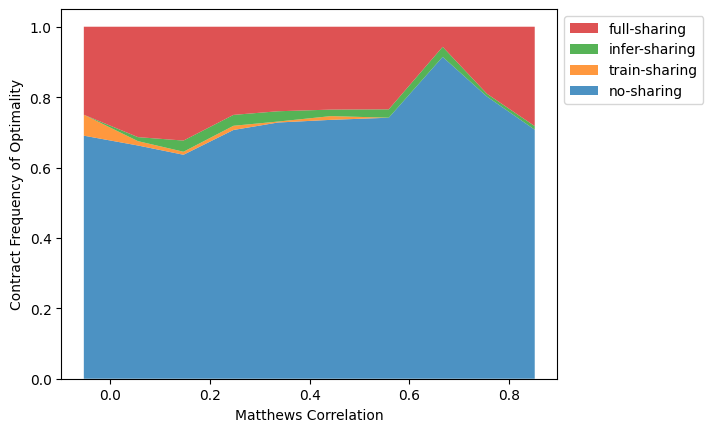}
\caption{Smoothed Contract Optimality vs. Firm Models' Correlation
}
\label{fig:smoothed-freq-vs-corr}
\end{figure}

Overall, the results reiterate the findings we detail in Section~\ref{sec:simulation}. A few surprising aspects to notice are the following: 
\begin{itemize}
\item Infer-sharing appears as a unique IRPO contract, contrary to our prediction. We believe that this is likely due to the full-sharing empirical decision that is based on the test data, where infer-sharing turns out to be better on the validation data. 

\item Train-sharing, which our example shows to emerge with high correlation rates, appears mostly with lower (or negative) correlations. It is reasonable, however, to expect that with negative correlation rates there could be examples of it as well. This is interesting as an indication of where train-sharing might be most relevant. 

\item In our implementation, the decision rules for no-sharing and infer-sharing are determined based on an assumption of independence (i.e. no correlation). We would thus expect the contracts to be more frequently optimal when this assumption is justified (low correlation rates). However, no-sharing is found to be somewhat more frequently optimal when the correlation is high. 

\end{itemize}

\end{document}